\documentclass[journal,twoside,web]{ieeecolor}
\usepackage{generic}
\usepackage{cite}
\usepackage{amsmath,amssymb,amsfonts}
\usepackage{graphicx}
\usepackage{algorithm,algorithmic}
\usepackage{textcomp}
\usepackage{epsfig}
\usepackage{epstopdf}
\usepackage{color}
\usepackage{float}
\usepackage{ntheorem}
\usepackage{calrsfs}
\DeclareMathAlphabet{\pazocal}{OMS}{zplm}{m}{n}
\usepackage{subcaption}
\makeatletter
\renewcommand*\env@matrix[1][*\c@MaxMatrixCols c]{%
  \hskip -\arraycolsep
  \let\@ifnextchar\new@ifnextchar
  \array{#1}}
\usepackage{arydshln}
\usepackage{mathtools}
\usepackage[T1]{fontenc} 
\usepackage{helvet}
\usepackage{soul}
\usepackage{courier}
\usepackage[ddmmyyyy, hhmmss]{datetime}

\definecolor{cadmiumgreen}{rgb}{0.0, 0.42, 0.24}
\usepackage[dvipsnames]{xcolor}
\newcommand{\revised}[1]{{\color{black}\selectfont #1}}
\newcommand{\revisedtwo}[1]{{\color{black}\selectfont #1}}

\usepackage{tikz}
\usetikzlibrary{tikzmark,patterns}
\usetikzlibrary{shapes,arrows,arrows.meta,chains,matrix,positioning,scopes,calc,patterns,shapes.geometric,intersections,backgrounds,fit,decorations.text}
\tikzset{
	block/.style = {draw, rectangle,
		minimum height=1cm,
		minimum width=2cm},
	input/.style = {coordinate,node distance=1cm},
	output/.style = {coordinate,node distance=4cm},
	arrow/.style={draw, -latex,node distance=2cm},
	pinstyle/.style = {pin edge={latex-, black,node distance=2cm}},
	sum/.style = {draw, circle, node distance=1cm},
}
\usetikzlibrary{shapes,arrows}
\usetikzlibrary{decorations.markings,decorations.pathmorphing,
	decorations.pathreplacing}
\usepackage[american voltages]{circuitikz}
\usetikzlibrary{patterns,arrows,shapes,calc,positioning,calc,intersections}
\usetikzlibrary{decorations.markings}
\usetikzlibrary{arrows.meta,shapes.arrows}

\def\BibTeX{{\rm B\kern-.05em{\sc i\kern-.025em b}\kern-.08em
    T\kern-.1667em\lower.7ex\hbox{E}\kern-.125emX}}
\usepackage[hidelinks]{hyperref}

\newtheorem{theorem}{Theorem}

\newtheorem{lemma}{Lemma}
\newtheorem{prop}{Proposition}
\newtheorem{assumption}{Assumption}
\newtheorem{definition}{Definition}

\newtheorem{problem}{Problem}

\title{Dissipativity-Based Data-Driven Decentralized Control of Interconnected Systems}
\author{Taiki Nakano, Ahmed Aboudonia, Jaap Eising, Andrea Martinelli, Florian Dörfler and John Lygeros
\thanks{
	T. Nakano, A. Martinelli, F. Dörfler and J. Lygeros are with the Automatic Control Laboratory, ETH Zurich, 8092 Switzerland {(e-mails: \{\texttt{nakanot}, \texttt{andremar}, \texttt{doerfler}, \texttt{lygeros}\}\texttt{@control.ee.ethz.ch}).
	T. Nakano is also with the Learning and Dynamical Systems Group, Max Planck Institute for Intelligent Systems, 72076 T{\"u}bingen, Germany.
  A. Aboudonia is with the University of California Berkeley, 94720 CA, USA (email: \texttt{aboudonia@berkeley.edu}).
	J. Eising is with ENTEG, University of Groningen, 9747 AG Groningen, The Netherlands (email: \texttt{j.eising@rug.nl}).
	Research supported by the Max Planck ETH Center for Learning Systems, NCCR Automation and the SNF/FW Weave Project 200021E\char`_20397.}
	}
}

\begin{document}
\maketitle

\begin{abstract}
	We propose data-driven decentralized control algorithms for stabilizing interconnected \revised{discrete-time linear time-invariant} systems.
	We first derive a data-driven condition to synthesize a local controller that ensures the dissipativity of the local subsystems.
	Then, we propose data-driven decentralized stability conditions for the global system based on the dissipativity of each local system.
  Since both conditions take the form of linear matrix inequalities and are based on dissipativity theory, this yields a unified pipeline, resulting in a data-driven decentralized control algorithm.
	As a special case, we also consider stabilizing systems interconnected through diffusive coupling and propose a control algorithm.
	We validate the effectiveness and the scalability of the proposed control algorithms in numerical examples in the context of microgrids.
\end{abstract}


\section{Introduction}
\label{sec:introduction}
\IEEEPARstart{C}{ontrol} of interconnected systems is an active research field
with applications in various cyber-physical domains including power systems, mobile autonomous robots, financial networks and the Internet.
These systems take the form of a set of local subsystems coupled by an interconnection structure, forming a global system.
Instead of analyzing the global system as a whole, we can often exploit the local properties of the subsystems to make the analysis far more tractable. 
Moreover, when designing controllers for the subsystems, the global interconnection structure is often not available.
This motivates the use of decentralized control methods, where the controllers are synthesized within each local subsystem and not in a centralized fashion~\cite{bakule2008decentralized,siljak2013decentralized,bakule2014decentralized}.

One particularly useful control design tool for decentralized control is passivity theory~\cite{van2000l2}, and more generally, dissipativity theory~\cite{willems1972dissipative1}. Dissipativity theory is a method to model energy exchange between subsystems and is used for the analysis and control synthesis of interconnected systems~\cite{chopra2006passivity,arcak2016networks}, including electrical networks and mechanical systems~\cite{ortega1998passivity}.
Although dissipativity theory is mainly developed for continuous-time systems,
discrete-time systems have also been studied~\cite{byrnes1993discrete,byrnes1994losslessness,navarro2002dissipativity,moreschini2024generalized},
also in the context of interconnected systems~\cite{mccourt2012stability,aboudonia2021passivity,martinelli2023interconnection}.
The developed control schemes, however, assume the availability of a model of the system.

The use of data in control design has been growing in recent years, fostered by unprecedented sensing availability, data storage capabilities and widespread success over multiple domains.
One typical strategy to control unknown dynamical systems by using data
is to first conduct system identification~\cite{ljung1998system}, then use the resulting model to design a controller~\cite{gevers2005identification}.
In contrast, direct data-driven control aims to directly synthesize a controller from collected data, circumventing the intermediate step of system identification.
The fundamental lemma~\cite{willems2005note} originating in behavioral systems theory~\cite{willems1997introduction} has recently inspired extensive research on direct data-driven analysis and synthesis methods~\cite{markovsky2008data,yang2015data,coulson2019data}.
Various recent efforts have been devoted to data-driven state-feedback control synthesis~\cite{de2019formulas,berberich2020data,van2020data} and data-driven dissipativity analysis~\cite{maupong2017lyapunov,koch2021provably,rosa2021one,van2022data}.
Fewer works have focused on making use of dissipativity theory to design data-driven state-feedback controllers~\cite{nguyen2024synthesis,kristovic2024data,tanaka2024algebraic} and only for individual systems.
There has also been research on data-driven control of networked systems \revised{\cite{baggio2021data,wang2023data,celi2023distributed,akbarzadeh2024data,liao2024data}}.
\revised{However, none of these works simultaneously handles an unknown interconnection structure and noise-corrupted data.
Handling these aspects simultaneously is challenging, since the local controller must be robust to noise in the local data while guaranteeing global stability for all interconnection structures consistent with the noisy data.}

\revised{
	In the paper, we focus on discrete-time linear time-invariant interconnected systems, where the models of the local subsystem and the interconnection structure are both unknown, and the data is corrupted by noise.
	By treating the supply rate matrices as joint decision variables in both the local dissipativity condition and the global stability condition, we propose a unified data-driven decentralized algorithm that synthesizes local controllers with global stability guarantees, without requiring any global data or centralized coordination.
}
The main contributions of the paper are as follows:
\begin{itemize}
	\item We derive a data-driven linear matrix inequality~(LMI) condition to compute a state-feedback controller that makes each subsystem dissipative;
	\item We propose data-driven decentralized LMI conditions that ensure the stability of the global system based on the dissipativity of each subsystem;
	\item By combining the previous, we obtain local control design algorithms that are data-driven, decentralized, and stabilize the global system. 
	\item Through numerical examples, we show that \revised{the} algorithms require less data and less computation time compared to a centralized data-driven control method.
\end{itemize} 

In Section \ref{sec:preliminaries}, we introduce preliminaries, while in Section \ref{sec:problem_formulation}, we present the problem formulation that we tackle in the paper.
In Section \ref{sec:solution-to-problem-1}, we propose a data-driven decentralized control algorithm for the considered class of interconnected systems.
In Section \ref{sec:control-design-for-unknown-diffusive-coupling}, we propose an algorithm for a specific class of interconnected systems with diffusive coupling.
In Section \ref{sec:numerical_examples}, we showcase some relevant numerical examples by considering a microgrid control case study.
Conclusions and directions for future work are outlined in Section \ref{sec:conclusions_and_discussions}.

\section{Preliminaries} \label{sec:preliminaries}
\subsection{Notation and matrix theory}
The $n\times m$ zero matrix and the $n\times n$ identity matrix are denoted by $0_{n \times m}$ and $I_n$ respectively;
the subscripts are omitted when the dimension is clear from the context.
Block diagonal matrices are denoted by $\operatorname{diag}(A_1,\dots,A_k)$.
The Moore-Penrose pseudoinverse of a matrix $M \in \mathbb{R}^{m \times n}$ is denoted by $M^\dagger \in \mathbb{R}^{n \times m}$.
The set of symmetric $n\times n$ matrices is denoted by $\mathbb{S}^n$.
We denote $A>0$  ($\geq 0$, $<0$, $\leq 0$)  when a symmetric matrix $A \in \mathbb{S}^n$ is positive definite (positive semidefinite, negative definite, negative semidefinite, respectively).
The kernel of a matrix $A$ is denoted by $\operatorname{ker}(A)$.
The inertia of a symmetric matrix $A$ is denoted by $\operatorname*{In}(A) = (\rho_-, \rho_0, \rho_+)$,
where $\rho_-,\rho_0$ and $\rho_+$ are the number of negative, zero and positive eigenvalues of $A$,
respectively.
We indicate objects that can be inferred by symmetry with the symbol ``$\bullet$''.

For a symmetric matrix $\Pi \in \mathbb{S}^{q+r}$, we define the following sets induced by a Quadratic Matrix Inequality (QMI) by
\begin{subequations}
  \begin{eqnarray}
      \pazocal{Z}_r(\Pi) \coloneqq  \left\{ Z \in \mathbb{R}^{r \times q} : 
    \begin{bmatrix}
      I_q \\ Z 
    \end{bmatrix}^\top 
    \Pi 
    \begin{bmatrix}
      I_q \\ Z 
    \end{bmatrix}
    \geq 0 \right\}, \label{eq:def-qmi} \\
    \pazocal{Z}^+_r(\Pi) \coloneqq  \left\{ Z \in \mathbb{R}^{r \times q} : 
    \begin{bmatrix}
      I_q \\ Z 
    \end{bmatrix}^\top 
    \Pi 
    \begin{bmatrix}
      I_q \\ Z 
    \end{bmatrix}
    > 0 \right\}. \label{eq:def-strict-qmi}
  \end{eqnarray}
\end{subequations}

We consider the partition of $\Pi$ as
$
  \Pi = 
  \left[
  \begin{smallmatrix}
    \Pi_{11} & \Pi_{12} \\
    \Pi_{21} & \Pi_{22}
  \end{smallmatrix}
  \right]
$
and define the set $\mathbf{\Pi}_{q,r}$ as
 \begin{equation}
  \begin{aligned}
    \mathbf{\Pi}_{q,r} \coloneqq 
  \left\{ \Pi \in \mathbb{S}^{q+r} : 
  \Pi_{22} \leq 0,~\Pi_{11} - \Pi_{12} \Pi_{22}^\dagger \Pi_{21} \geq 0, 
  \right. \\ \left.
  \operatorname*{ker}(\Pi_{22}) \subseteq  \operatorname*{ker}(\Pi_{12}) \right\}.
  \end{aligned}
 \end{equation}

\begin{prop}[Matrix S-lemma {\cite[Theorem 4.7]{van2023quadratic}}] \label{prop:matrix-S-lemma}
  Consider $M,N \in \mathbb{S}^{q+r}$.
  If there exists a scalar $\alpha \geq 0$ such that $M - \alpha N \geq 0$, 
  then $\pazocal{Z}_r(N) \subseteq \pazocal{Z}_r(M)$.
  Moreover, if $N \in \mathbf{\Pi}_{q,r}$ and $N$ has at least one positive eigenvalue, then $\pazocal{Z}_r(N) \subseteq \pazocal{Z}_r(M)$ if and only if there exists $\alpha \geq 0$ such that $M - \alpha N \geq 0$.
\end{prop}

For a set $\pazocal{S} \subseteq \mathbb{R}^{r \times q}$, we define $\pazocal{S}^\top \coloneqq \{ Z^\top \mid Z \in \pazocal{S} \}$.
\begin{prop}[Dual QMI~{{\cite[Proposition 3.1]{van2023quadratic}}}] \label{prop:dualization-lemma}
  Let $\Pi \in \mathbb{S}^{q+r}$ be invertible and let $\pazocal{Z}_r(\Pi)$ be nonempty.
  Then, $\left( \pazocal{Z}_r(\Pi) \right)^\top = \pazocal{Z}_q(\hat{\Pi})$ for
  $\hat{\Pi} \coloneqq 
  \left[
    \begin{smallmatrix}
      0 & - I_r \\
      I_q & 0 
    \end{smallmatrix}
  \right]
    \Pi^{-1}
    \left[
    \begin{smallmatrix}
      0 & - I_q \\
      I_r & 0
    \end{smallmatrix}
    \right]
    $.
\end{prop}

\subsection{Dissipativity theory}
\begin{definition} \label{def:dissipative-systems}
  The system 
  \begin{equation} \label{eq:lti-system}
    \begin{aligned}
      x(t+1) &= A x(t) + B v(t), \\
      y(t) &=  C x(t) + D v(t),
  \end{aligned}
  \end{equation}
  with $x(t) \in \mathbb{R}^{n}$, $v(t) \in \mathbb{R}^q$ and $y(t) \in \mathbb{R}^p$ is said to be dissipative with respect to the supply rate
  $
    s(v(t),y(t)) = 
    \left[
    \begin{smallmatrix}
      v(t) \\ y(t)  
    \end{smallmatrix}\right]^\top 
    S 
    \left[
    \begin{smallmatrix}
      v(t) \\ y(t)  
    \end{smallmatrix}\right]
  $
  with $S \in \mathbb{R}^{(q+p) \times (q+p)}$,
  if there exists a {storage function}
  $
    V(x) = x^\top P x
  $
  with $P > 0$,
  such that for all $x(t) \in \mathbb{R}^n$ and $v(t) \in \mathbb{R}^q$,
  \begin{equation*}
      V(x(t+1))-V(x(t)) \leq s(v(t),y(t)).
  \end{equation*}
\end{definition}

\revised{Dissipativity captures the relation that the change of stored energy in the system does not exceed the supplied energy.}
We also note that quadratic supply rates include important system-theoretic properties such as passivity by taking
$S = 
\left[ 
  \begin{smallmatrix}
    0 & \frac{1}{2}I \\
    \frac{1}{2}I & 0
  \end{smallmatrix}
\right]$
and finite-gain $\pazocal{L}_2$-stability with gain less than or equal to $\gamma >0$
by taking $
  S = 
\left[
  \begin{smallmatrix}
    \gamma^2 I_q & 0 \\
    0 & - I_p
  \end{smallmatrix}
\right]$.
It can also be shown that for quadratic supply rates one can consider a quadratic storage function without loss of generality.
Moreover, the storage function matrix $P$ can be assumed to be positive definite without loss of generality as long as the pair $(C,A)$ is observable \cite{burohman2023data}.

Whether a system is dissipative or not can be verified by the feasibility of a linear matrix inequality (LMI).
\revised{In this paper, we work with quadratic supply rates that satisfy $\operatorname*{In}(S) = (p, 0, q)$, and recall an LMI formulation of Definition~\ref{def:dissipative-systems}.}
  \revised{
  \begin{prop}[Dual dissipativity LMI~{\cite[Proposition 2]{van2022data}}] \label{prop:dual-dissipativity-lmi}
    Assume that the matrix~$S$ has inertia~$\operatorname*{In}(S) = (p, 0, q)$ and $P  > 0$.
    Then, system \eqref{eq:lti-system} is dissipative with respect to the supply rate $s(v(t),y(t)) =\left[\begin{smallmatrix}
        v(t) \\ y(t)  
      \end{smallmatrix}\right]^\top 
      S
      \left[
        \begin{smallmatrix}
          v(t) \\ y(t)  
        \end{smallmatrix}
      \right]$, if and only if the following LMI holds:
      \begin{align*}
        \setlength\arraycolsep{0.2pt}
      &
      \begin{bmatrix}
        I & 0 \\ 
        A^\top & C^\top 
      \end{bmatrix}^\top 
      \begin{bmatrix}
        P^{-1} & 0 \\
        0 & - P^{-1}
      \end{bmatrix}
      \begin{bmatrix}
        I & 0 \\ 
        A^\top & C^\top
      \end{bmatrix} \\
      + &
      \begin{bmatrix}
        0 & I \\ 
        B^\top & D^\top 
      \end{bmatrix}^\top 
      \begin{bmatrix}
        0 & - I \\
        I & 0 
      \end{bmatrix}
      S^{-1}
      \begin{bmatrix}
        0 & - I \\
        I & 0
      \end{bmatrix} 
      \begin{bmatrix}
        0 & I \\ 
        B^\top & D^\top
      \end{bmatrix} \geq 0 
    \end{align*}
  \end{prop}
  }

  \revised{Note that while the supply rate is expressed in 
  terms of $S$, the LMI condition involves $S^{-1}$, which 
  is well-defined under the inertia assumption 
  $\operatorname*{In}(S) = (p, 0, q)$.}
  As introduced in~\cite[Assumption (A1)]{van2022data}, 
  the assumption $\operatorname*{In}(S) = (p, 0, q)$ is satisfied for system-theoretic properties including passivity and finite-gain $\pazocal{L}_2$ stability.

\section{Problem Formulation}
\label{sec:problem_formulation}
\subsection{System class}
We consider a class of systems comprising $k$ heterogeneous square linear time-invariant subsystems interconnected with each other.
The dynamics of the $i$-th subsystem are
\begin{equation} \label{eq:dynamics-of-the-ith-subsystem-with-noise}
	\begin{aligned}
		x_i(t+1) &= A_i x_i(t) + B_{1,i} u_i(t) + B_{2,i} v_i(t) + w_{1,i}(t) , \\
		y_i(t)   &= C_i x_i(t) + D_{1,i} u_i(t) + D_{2,i} v_i(t) + w_{2,i}(t),
	\end{aligned}
\end{equation}
where $x_i(t) \in \mathbb{R}^{n_i}$ is the state, $u_i(t) \in \mathbb{R}^{m_i}$ is the {control input}, $v_i(t) \in \mathbb{R}^{p_i}$ is the {interconnection input}, which is an exogenous input affected by other subsystems, $w_{1,i}(t) \in \mathbb{R}^{n_i}$ is the process noise, $w_{2,i}(t) \in \mathbb{R}^{p_i}$ is the measurement noise and $y_i(t) \in \mathbb{R}^{p_i}$ is the output.
The output affects the interconnection inputs through a linear interconnection structure
\begin{equation} \label{eq:interconnection-via-linear-relation-with-noise}
  v_i(t) = \sum_{j \in \pazocal{N}_i} M_{ij} y_j(t) + \xi_i(t),
\end{equation}
for each $i\in\{1,\dots,k\}$,
where $\xi_i(t) \in \mathbb{R}^{p_i}$ is the interconnection noise and $M_{ij} \in \mathbb{R}^{p_i \times p_j}$.
Here, $\pazocal{N}_i$ is the neighborhood set of the $i$-th subsystem that contains the indices of all subsystems whose outputs affect the $i$-th subsystem.
We assume for simplicity that $i \in \pazocal{N}_i$ holds.

We denote the global interconnection matrix $M$ by
\begin{equation*}
	M  = 
	\begin{bmatrix} 
		M_{11} & \cdots & M_{1k} \\
		\vdots & \ddots & \vdots \\
		M_{k1} & \cdots & M_{kk}
	\end{bmatrix}
	\in \mathbb{R}^{p \times p} .
\end{equation*}
From \eqref{eq:interconnection-via-linear-relation-with-noise}, we set $M_{ij} = 0$ for all $j \notin \pazocal{N}_i$.
We also assume that the interconnection is symmetric, that is, $M = M^\top$ holds.
This is often not restrictive, for example in cases where the interconnection relation can be characterized by some symmetric relations between subsystems.
For example, bidirectional diffusive coupling, a specific type of interconnection structure, satisfies this assumption.
\revised{
Additionally, we assume that $I - 
\operatorname{diag}(D_{2,1},\ldots,D_{2,k}) M$ is 
invertible to avoid algebraic loops in the interconnected 
system.}

We denote the elements of the set $\pazocal{N}_i$ as 
$\pazocal{N}_i = \{i\} \cup \{\sigma(1),\dots,\sigma(|\pazocal{N}_i|-1)\}$ and define $\tilde{p}_i \coloneqq \sum_{j \in \pazocal{N}_i}p_j$ and
\begin{equation*}
  \begin{aligned}
    \tilde{y}_i(t) &\coloneqq \begin{bmatrix}
			y_i (t) \\
      y_{\sigma(1)} (t) \\
      \vdots \\
      y_{\sigma(|\pazocal{N}_i|-1)} (t) 
    \end{bmatrix} \in \mathbb{R}^{\tilde{p}_i}, \\
    \tilde{M}_i^\mathrm{r} &\coloneqq 
    \begin{bmatrix}
			M_{ii} & M_{i\sigma(1)} & \cdots & M_{i\sigma(|\pazocal{N}_i|-1)}
    \end{bmatrix} \in \mathbb{R}^{p_i \times \tilde{p}_i} .
  \end{aligned}
  \end{equation*}
Then, the symmetric interconnection structure \eqref{eq:interconnection-via-linear-relation-with-noise} can be written as 
\begin{equation} \label{eq:interconnection-output-local-relation}
  v_i(t) = \tilde{M}^\mathrm{r}_i \tilde{y}_i(t) + \xi_i(t).
\end{equation}

As a special case of the symmetric interconnection structure, we also consider diffusive coupling~\cite[Chapter 8]{bullo2018lectures}, an interconnection structure seen in many applications whenever potential variables (for example voltages in the microgrid case) are used as internal variables.
The interconnection structure is a diffusive coupling, if there exist scalars $\{a_{ij}\}_{j \in \pazocal{N}_i}$ called {weights}, such that
\begin{equation} \label{eq:diffusive-coupling}
  v_i(t) = \sum_{j \in \pazocal{N}_i} a_{ij} (y_j(t) - y_i(t)) + \xi_i(t) ,
\end{equation}
where $a_{ij} > 0$ for $j \in \pazocal{N}_i \setminus \{i\}$ and $a_{ii} = 0$.
We also define the weighted degree as $d_i \coloneqq \sum_{j \in \pazocal{N}_i} a_{ij}$.

\subsection{Data collection \revised{and available information}}
\revised{ 
Each subsystem has the knowledge of the neighbor set $\pazocal{N}_i$, but not the interconnection weights $M_{ij}$.
Based on this knowledge, each subsystem collects two types of data independently, requiring no global coordination.
First, the local data is collected independently at each subsystem by measuring the local state, input, interconnection input, and output trajectories.
Second, the interconnection data is collected from the interconnection relation~\eqref{eq:interconnection-output-local-relation} by measuring the local interconnection input and the output of the neighboring subsystems 
in $\pazocal{N}_i$.
}

\subsubsection*{Local data}
We assume that we do {not} know the system matrices of each subsystem in \eqref{eq:dynamics-of-the-ith-subsystem-with-noise},
but that we can collect data trajectories $\{x_i(t)\}_{t=0}^{N_i}$, $\{u_i(t)\}_{t=0}^{N_i-1}$, $\{v_i(t)\}_{t=0}^{N_i-1}$ and $\{y_i(t)\}_{t=0}^{N_i-1}$ of lengths $N_i$.
We define the following matrices
\begin{equation} \label{eq:subsystem-data-matrices-definition}
	\begin{aligned}
		X_i &\coloneqq
		\begin{bmatrix}
			x_i(0) & x_i(1) & \cdots & x_i(N_i-1)
		\end{bmatrix},   \\
		X_i^+ &\coloneqq
		\begin{bmatrix}
			x_i(1) & x_i(2) & \cdots & x_i(N_i)
		\end{bmatrix}, \\
		U_i &\coloneqq \begin{bmatrix}
			u_i(0) & u_i(1) & \cdots & u_i(N_i-1)
		\end{bmatrix}, \\
		V_i &\coloneqq \begin{bmatrix}
			v_i(0) & v_i(1) & \cdots & v_i(N_i-1)
		\end{bmatrix}, \\
		Y_i &\coloneqq \begin{bmatrix}
			y_i(0) & y_i(1) & \cdots & y_i(N_i-1)
		\end{bmatrix}, \\
		W_i &\coloneqq \begin{bmatrix}
			w_{1,i}(0) & w_{1,i}(1) & \cdots & w_{1,i}(N_i-1) \\
      w_{2,i}(0) & w_{2,i}(1) & \cdots & w_{2,i}(N_i-1)
		\end{bmatrix}, \\
	\end{aligned}
\end{equation}
and denote the {local data} $\mathcal{D}_i \coloneqq \{ X_i, X_i^+, U_i, V_i, Y_i\}$.
We note that, while the first five matrices in \eqref{eq:subsystem-data-matrices-definition} can be constructed, the last matrix cannot, as it involves the unmeasured noise trajectories $\{w_{1,i}(t)\}_{t=0}^{N_i-1}$ and $\{w_{2,i}(t)\}_{t=0}^{N_i-1}$.
We assume, however, that the noise trajectories give rise to a matrix $W_i$ that satisfies a known QMI.
\begin{assumption} \label{assumption:noise-matrices-quadratic-bound}
	The matrix $W_i$ is unknown but satisfies
	\begin{equation} \label{eq:prob-formulation-quadratic-noise-bound}
		W_i^\top \in \pazocal{Z}_{N_i}(\Phi_i),
	\end{equation}
	with a known matrix $
	\Phi_i =
  \left[  
    \begin{smallmatrix}
      \Phi_{11,i} & \Phi_{12,i} \\
      \Phi_{12,i}^\top & \Phi_{22,i}
    \end{smallmatrix} 
  \right]
	\in \mathbf{\Pi}_{n_i+p_i, N_i}$.
\end{assumption}

Noise models based on QMIs similar to Assumption \ref{assumption:noise-matrices-quadratic-bound} are also used in works that consider data-driven dissipativity~\cite{berberich2020robust, koch2021provably, van2022data}, and can impose bounds on sequences, the maximal singular value of~$W_i$ and individual noise samples by an appropriate choice of the matrix~$\Phi_i$.
For instance, by setting $(\Phi_{11,i},\Phi_{12,i},\Phi_{22,i}) = (N_i \varepsilon I, 0, -I)$, $\pazocal{Z}_{N_i}(\Phi_i)$ contains all vectors such that~$\left\|  
\left[  
  \begin{smallmatrix}
    w_{1, i}(t) \\ w_{2, i}(t)
\end{smallmatrix}
\right]  
 \right\|_2^2 \leq \varepsilon$ is satisfied for all~$t = 0,\dots,N_i-1$.

\subsubsection*{Interconnection data}
For each subsystem, we also assume that the interconnection matrix $M$ is unknown for each subsystem but that $\pazocal{N}_i$ is known and the trajectories $\{v_i(t)\}_{t=0}^{\tilde{N}_{i} - 1}$ and $\{\tilde{y}_i(t)\}_{t=0}^{\tilde{N}_{i} - 1}$ of lengths $\tilde{N}_i$ can be collected from the interconnection relation \eqref{eq:interconnection-output-local-relation}.
We define the following matrices
\begin{equation*}
  \begin{aligned}
    \tilde{V}_i &\coloneqq 
  \begin{bmatrix}
    v_i(0) & \cdots & v_i(\tilde{N}_{i} - 1)
  \end{bmatrix}, \\
  \tilde{Y}_i &\coloneqq 
  \begin{bmatrix}
    \tilde{y}_i(0) & \cdots & \tilde{y}_i(\tilde{N}_{i} - 1)
  \end{bmatrix}, \\
  \Xi_i &\coloneqq 
  \begin{bmatrix}
    \xi_i(0) & \cdots & \xi_i(\tilde{N}_{i}-1)
  \end{bmatrix},
  \end{aligned}
\end{equation*}
and the  interconnection data $\tilde{\mathcal{D}}_i \coloneqq \{\tilde{V}_i,\tilde{Y}_i\}$.
While the matrix $\Xi_i$ cannot be constructed from data, we assume that it is bounded in the form of a QMI.

\begin{assumption} \label{assumption:interconnection-noise-matrices-QMI-local}
	The matrix 
  $
  \Xi_i
  $ satisfies
  \begin{equation} \label{eq:interconnection-matrix-relation}
    \Xi_i^\top \in \pazocal{Z}_{\tilde{N}_{i}} (\Psi_i),
  \end{equation}
	with a known matrix
  $\Psi_i =
  \left[
    \begin{smallmatrix}
      \Psi_{11,i} & \Psi_{12,i} \\
      \Psi_{12,i}^\top & \Psi_{22,i}
    \end{smallmatrix}
  \right]
   \in \mathbf{\Pi}_{p_i, \tilde{N}_{i}}$.
\end{assumption}

\subsection{Control objective} \label{subsec:control-objective}
Consider local state-feedback controllers for each subsystem \eqref{eq:dynamics-of-the-ith-subsystem-with-noise}, that is, $u_i = K_i x_i$, with $K_i \in \mathbb{R}^{m_i \times n_i}$. 
Using this controller, the closed-loop dynamics of the $i$-th subsystem can be written as 
\begin{equation} \label{eq:closed-loop-system-ith-with-noise}
	\begin{aligned}
		x_i(t+1) &= (A_i+B_{1, i} K_i) x_i(t) + B_{2,i} v_i(t) + w_{1,i}(t), \\
    y_i(t) &= (C_i+D_{1, i} K_i) x_i(t) + D_{2,i} v_i(t) + w_{2,i}(t),
	\end{aligned}
\end{equation}
leading to the global closed-loop dynamics
\begin{equation} \label{eq:global-dynamics-with-noise}
  \begin{aligned}
    x(t+1) &= (A+B_1K) x(t) + B_2 v(t) + w_1(t), \\
    y(t) &= (C+D_1K) x(t) + D_2 v(t) + w_2(t), \\ 
    v(t) &= M y(t) + \xi(t),
  \end{aligned}
\end{equation}
where
$x(t) = \begin{bmatrix}
  x_1(t)^\top & \cdots & x_k(t)^\top
\end{bmatrix}^\top \in \mathbb{R}^n$
(similarly for 
$v(t) \in \mathbb{R}^p$, $u(t) \in \mathbb{R}^m$, $y(t)  \in \mathbb{R}^p$, $w_1(t) \in \mathbb{R}^{n}$, $w_2(t) \in \mathbb{R}^{p}$ and $\xi(t) \in \mathbb{R}^{p}$),
and  
$A = \operatorname{diag}(A_1,\dots,A_k) \in \mathbb{R}^{n \times n}$
(similarly for 
\revised{$B_1 \in \mathbb{R}^{n \times m}$,
$B_2 \in \mathbb{R}^{n \times p}$},
$C \in \mathbb{R}^{p \times n}$,
$D_1 \in \mathbb{R}^{p \times m}$,
$D_2 \in \mathbb{R}^{p \times p}$ and
$K \in \mathbb{R}^{m \times n}$).
\revisedtwo{The} aim is to design local state-feedback controllers $u_i = K_i x_i$ for all $i\in\{1,\dots,k\}$ such that the nominal global closed-loop system
\begin{equation} \label{eq:global-dynamics}
  \begin{aligned}
    x(t+1) &= (A+B_1K) x(t) + B_2 v(t), \\
    y(t) &= (C+D_1K) x(t) + D_2 v(t), \\ 
    v(t) &= M y(t).
  \end{aligned}
\end{equation}
is asymptotically stable. 
From \cite{jiang2001input}, this ensures the input-to-state stability of the global closed-loop system \eqref{eq:global-dynamics-with-noise} with respect to the noise vectors $w_1(t)$, $w_2(t)$ and $\xi(t)$.
System \eqref{eq:global-dynamics} is illustrated in Figure~\ref{fig:global-closed-loop}.

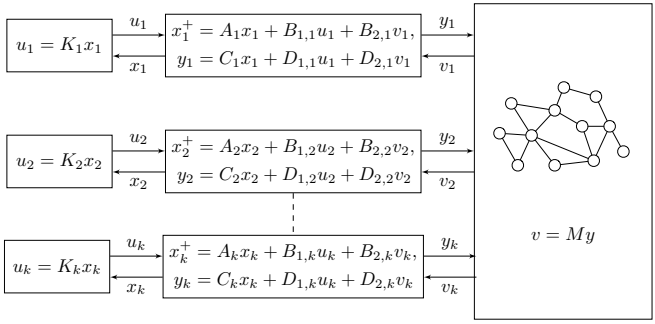
\begin{figure}
	\begin{center}
		\resizebox{0.48\textwidth}{!}{
			\begin{tikzpicture}[auto, node distance=1cm,>=latex']
				\node [block, name=plant] (plant) {$\begin{aligned}
					x_1^+ &= A_1 x_1 + B_{1,1} u_1 + B_{2,1} v_1 , \\
					y_1   &= C_1 x_1 + D_{1,1} u_1 + D_{2,1} v_1
  \end{aligned}$};
				\node [block, below=1 of plant] (plant2) {$\begin{aligned}
					x_2^+ &= A_2 x_2 + B_{1,2} u_2 + B_{2,2} v_2 , \\
					y_2   &= C_2 x_2 + D_{1,2} u_2 + D_{2,2} v_2
  \end{aligned}$};
				\node [block, below=3 of plant] (plant3) {$\begin{aligned}
					x_k^+ &= A_k x_k + B_{1,k} u_k + B_{2,k} v_k , \\
					y_k   &= C_k x_k + D_{1,k} u_k + D_{2,k} v_k
  \end{aligned}$};
				\node [block, left= 1 of plant] (K1) {$u_1 = K_1x_1$};
				\node [block, left= 1 of plant2] (K2) {$u_2 = K_2x_2$};
				\node [block, left= 1 of plant3] (K3) {$u_k = K_kx_k$};
				\draw [dashed] (plant2) -- (plant3);
				\node [block, right =1 of plant2,minimum height=6cm] (network) {\begin{tabular}{c} \includegraphics[width=0.15\textwidth, trim = 3.2cm 0 0 1, clip]{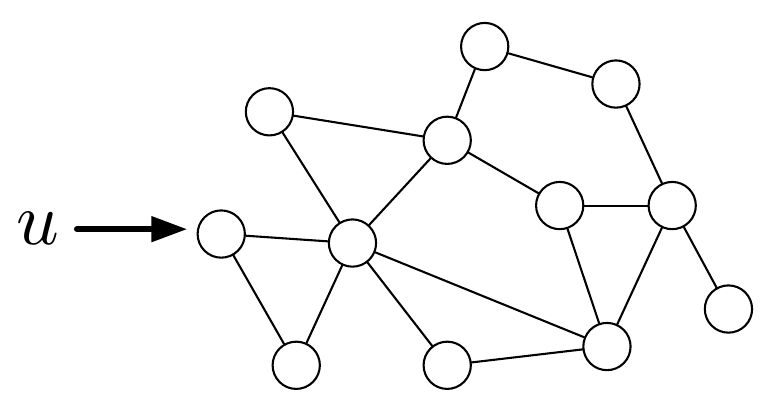}
						\\
						\vspace{1em}
						\\
						$v= M y$
				\end{tabular}};
				\draw [draw,->] ($(plant.east)+(0,0.2)$) -- node {$y_1$} ($(plant.east)+(1,0.2)$);
				\draw [draw,->] ($(plant.east)+(1,-0.2)$) -- node {$v_1$} ($(plant.east)+(0,-0.2)$);
				\draw [draw,->] ($(plant2.east)+(0,0.2)$) -- node {$y_2$} ($(plant2.east)+(1,0.2)$);
				\draw [draw,->] ($(plant2.east)+(1,-0.2)$) -- node {$v_2$} ($(plant2.east)+(0,-0.2)$);
				\draw [draw,->] ($(plant3.east)+(0,0.2)$) -- node {$y_k$} ($(plant3.east)+(1,0.2)$);
				\draw [draw,->] ($(plant3.east)+(1,-0.2)$) -- node {$v_k$} ($(plant3.east)+(0,-0.2)$);
				\draw [draw,->] ($(plant.west)-(0,0.2)$) -- node {$x_1$} ($(plant.west)-(1,0.2)$);
				\draw [draw,->] ($(plant.west)-(1,-0.2)$) -- node {$u_1$} ($(plant.west)-(0,-0.2)$);
				\draw [draw,->] ($(plant2.west)-(0,0.2)$) -- node {$x_2$} ($(plant2.west)-(1,0.2)$);
				\draw [draw,->] ($(plant2.west)-(1,-0.2)$) -- node {$u_2$} ($(plant2.west)-(0,-0.2)$);
				\draw [draw,->] ($(plant3.west)-(0,0.2)$) -- node {$x_k$} ($(plant3.west)-(1,0.2)$);
				\draw [draw,->] ($(plant3.west)-(1,-0.2)$) -- node {$u_k$} ($(plant3.west)-(0,-0.2)$);
		\end{tikzpicture}}
	\end{center}
	\caption{Illustration of the nominal global closed-loop system \eqref{eq:global-dynamics}.
	We aim at designing $\{K_i\}_{i = 1}^k$ such that \eqref{eq:global-dynamics} is asymptotically stable
	on the basis of data.} 	
	\label{fig:global-closed-loop}			
\end{figure}

\begin{problem} \label{prb:decentralized-unknown}
  Assume that for the $i$-th subsystem, where $i \in \{1,\dots, k\}$, 
	the data $\mathcal{D}_i$, $\tilde{\mathcal{D}}_i$ and the neighbor set $\pazocal{N}_i$ is available and Assumptions \ref{assumption:noise-matrices-quadratic-bound} and \ref{assumption:interconnection-noise-matrices-QMI-local} hold.
	Compute a state-feedback gain $K_i$
  such that the nominal global closed-loop system~\eqref{eq:global-dynamics} is asymptotically stable.
\end{problem}

In addition, we investigate systems under diffusive coupling. Here, we assume for simplicity that $p_i = 1$ for all $i \in \{1,\dots, k \}$; similar results can be obtained when the input and output dimensions are greater than one.
Problem \ref{prb:decentralized-unknown} then reduces to:
\begin{problem} \label{prb:diffusive-coupling}
  Assume that $p_i = 1$ for all $i \in \{1,\dots, k\}$
  and the interconnection structure is a diffusive coupling \eqref{eq:diffusive-coupling}.
	Assume that for the $i$-th subsystem, where $i \in \{1,\dots, k\}$, 
	the data $\mathcal{D}_i$, $\tilde{\mathcal{D}}_i$ and the neighbor set $\pazocal{N}_i$ is available and Assumptions \ref{assumption:noise-matrices-quadratic-bound} and \ref{assumption:interconnection-noise-matrices-QMI-local} hold.
	Compute a state-feedback gain $K_i$ such that the nominal global closed-loop system~\eqref{eq:global-dynamics} is asymptotically stable.
\end{problem}

Note that in Problems \ref{prb:decentralized-unknown} and \ref{prb:diffusive-coupling}, we consider both ``decentralized control'' in the sense that each controller of the system does not communicate with other controllers, and ``decentralized design'' in the sense that the control design does not require global data or information.

\section{Dissipativity-Based Data-Driven Decentralized Control} \label{sec:solution-to-problem-1}
To address Problem~\ref{prb:decentralized-unknown}, we adopt a two-step approach: we first develop a local dissipativity condition for each subsystem, then derive a decentralized stability condition based on the supply rates of each subsystem.
Both conditions are data-driven, decentralized, and described by LMIs, allowing us to combine them, leading to a pipeline for synthesizing data-driven decentralized controllers.

\subsection{Local Controller Design} \label{sec:local_controller_design}
We first derive a data-driven LMI condition to compute a state-feedback gain~$K_i$
that makes the $i$-th nominal closed-loop system
\begin{equation} \label{eq:closed-loop-system-local}
  \begin{aligned}
    x_i(t+1) &= ( A_i + B_{1,i} K_i ) x_i(t) + B_{2,i} v_i(t), \\
    y_i(t) &=  (C_i + D_{1,i} K_i ) x_i(t) + D_{2,i} v_i(t),
  \end{aligned}
\end{equation}
dissipative with respect to the supply rate
\begin{equation} \label{eq:supply-rate}
  \begin{aligned}
    s_i(v_i(t),y_i(t))=
    \begin{bmatrix}
      v_i(t) \\ y_i(t)  
    \end{bmatrix}^\top 
    \begin{bmatrix}
      H_i & G_i^\top \\
      G_i & F_i
    \end{bmatrix}^{-1}
    \begin{bmatrix}
      v_i(t) \\ y_i(t)  
    \end{bmatrix}.
  \end{aligned}
\end{equation}
Here, the matrices~${F}_i \in \mathbb{R}^{p_i \times p_i}$,~${G}_i \in \mathbb{R}^{p_i \times p_i}$ and~${H}_i \in \mathbb{R}^{p_i \times p_i}$ are such that
\begin{equation} \label{eq:inertia-condition}
  \operatorname*{In} \left( 
    \begin{bmatrix}
      H_i & G_i^\top \\
      G_i & F_i
    \end{bmatrix}^{-1}
  \right) 
  =
  (p_i, 0, p_i).
\end{equation}

Given the local data $\mathcal{D}_i = \{ X_i, X_i^+, U_i, V_i, Y_i\}$ constructed from the data trajectories in \eqref{eq:subsystem-data-matrices-definition}, and recalling that the trajectories are generated from the dynamics \eqref{eq:dynamics-of-the-ith-subsystem-with-noise}, we can show that
\begin{equation} \label{eq:local-data-matrices-relation}
    \begin{bmatrix}
      X_i^+ \\
      Y_i
    \end{bmatrix}
    = 
    \begin{bmatrix}
      A_i X_i + B_{1,i} U_i + B_{2,i} V_i \\
      C_i X_i + D_{1,i} U_i + D_{2,i} V_i
    \end{bmatrix}
    + W_i.
\end{equation}
Then, under the assumption on the noise matrix, the set of system matrices that are consistent with the local data is defined as
\begin{equation*} \label{eq:set-of-systems-consistent-with-data}
  \begin{split}
    \Sigma_i
  \coloneqq
  & \left\{ (A_i,B_{1,i},B_{2,i},C_i,D_{1,i},D_{2,i})
  \right. \\
  &
  \left.
  :\text{$\exists W_i$ s.t.~\eqref{eq:prob-formulation-quadratic-noise-bound} and \eqref{eq:local-data-matrices-relation} hold} \right\}.
  \end{split}
\end{equation*}

Since we cannot distinguish the true system matrices from the other matrices in $\Sigma_i$,
we aim to make system \eqref{eq:closed-loop-system-local} dissipative for all matrices in $\Sigma_i$, given the local data $\mathcal{D}_i$.
The following theorem provides a data-driven condition to design such a state-feedback controller.
\begin{theorem} \label{thm:local-controller-design}
  \revisedtwo{Consider the data $\mathcal{D}_i = \{ X_i, X_i^+, U_i, V_i, Y_i\}$ and assume that it satisfies Assumption~\ref{assumption:noise-matrices-quadratic-bound}.}
  Given the matrices $P_i>0,L_i,F_i,G_i,H_i$, define %
  \begin{equation} \label{eq:data-matrices-in-the-theorem}
    \begin{aligned}
      \hat{M}_i &\coloneqq 
    \left[\begin{array}{ccccc;{2pt/2pt}c}
      P_i & 0 & 0 & 0 & 0 & 0\\
      0 & - F_i & 0 & 0 & G_i & 0 \\
      0 & 0 & - P_i & - L_i^\top & 0 & 0\\ 
      0 & 0 & - L_i & 0 & 0 & L_i \\ 
      0 & G_i^\top & 0 & 0 & - H_i & 0 \\ \hdashline[2pt/2pt]
      0 & 0 & 0 & L_i^\top & 0  & P_i
    \end{array}
    \right],
      \\
    J_i &\coloneqq
    \left[
    \begin{array}{c;{2pt/2pt}c}
      I & 
      \begin{matrix}
        X_i^+ \\
      Y_i 
      \end{matrix}  \\ \hdashline[2pt/2pt]
      0 & \begin{matrix}
        - X_i \\
      - U_i \\
      - V_i
      \end{matrix}
    \end{array}
    \right]
    \Phi_i
    [\bullet]^\top, 
    \hat{N}_i \coloneqq 
    \left[
    \begin{array}{c;{2pt/2pt}c}
      J_i
      & 0 \\ \hdashline[2pt/2pt] 
      0 & 0_{n_i \times n_i}
    \end{array}
    \right],
  \end{aligned}
  \end{equation}\renewcommand{\arraystretch}{1}and assume that~$J_i$ has at least one positive eigenvalue.
  Then, there exists a state-feedback gain~$K_i$ such that for any matrices $(A_i,B_{1,i},B_{2,i},C_i,D_{1,i},D_{2,i}) \in \Sigma_i$, system~\eqref{eq:closed-loop-system-local} is dissipative with respect to the supply rate~\eqref{eq:supply-rate} that satisfies \eqref{eq:inertia-condition}, if and only if the LMI
  \begin{equation} \label{eq:local-controller-design-lmi}
    \hat{M}_i - \alpha_i \hat{N}_i \geq 0 
  \end{equation}
  with variables $P_i >0$, $L_i$ and $\alpha_i \geq 0$ holds.
  Moreover, if $P_i$ and $L_i$ satisfy \eqref{eq:local-controller-design-lmi}, then $K_i = L_i P_i^{-1}$ is such a state-feedback gain and $V_i(x_i) = x_i^\top P_i^{-1} x_i$ is the corresponding storage function.
\end{theorem}

\begin{proof}
  First, by substituting~\eqref{eq:local-data-matrices-relation} into~\eqref{eq:prob-formulation-quadratic-noise-bound} to eliminate~$W_i$
  and rearranging the resulting terms, we note that the QMI
  \renewcommand{\arraystretch}{1.15} %
  \begin{align} \label{eq:quadratic-bound-before-s-lemma}
    \setlength\arraycolsep{3pt}
    \left[
    \begin{array}{cc}
      I \\ \hdashline 
      \begin{matrix}
        A_i^\top & C_i^\top
      \end{matrix} \\
      \begin{matrix}
        B_{1,i}^\top & D_{1,i}^\top
      \end{matrix} \\
      \begin{matrix}
        B_{2,i}^\top & D_{2,i}^\top
      \end{matrix}
    \end{array}
    \right]^\top
    J_i 
    \left[
    \begin{array}{cc}
      I \\ \hdashline 
      \begin{matrix}
        A_i^\top & C_i^\top
      \end{matrix} \\
      \begin{matrix}
        B_{1,i}^\top & D_{1,i}^\top
      \end{matrix} \\
      \begin{matrix}
        B_{2,i}^\top & D_{2,i}^\top
      \end{matrix}
    \end{array}
    \right]
    &\geq 0 
  \end{align}
  \renewcommand{\arraystretch}{1}holds if and only if $(A_i,B_{1,i},B_{2,i},C_i,D_{1,i},D_{2,i}) \in \Sigma_i$.
  From Assumption \ref{assumption:noise-matrices-quadratic-bound}, \eqref{eq:data-matrices-in-the-theorem} and \eqref{eq:quadratic-bound-before-s-lemma}, one can show that
  $J_i \in \mathbf{\Pi}_{n_i+p_i, n_i+m_i+p_i}$.

  Next, we reformulate the dissipativity condition of system \eqref{eq:closed-loop-system-local} as a QMI.
  We define~$\hat{A}_i\coloneqq A_i+B_{1,i}K_i$ and $\hat{C}_i \coloneqq C_i + D_{1,i} K_i$.
  From \revised{Proposition \ref{prop:dual-dissipativity-lmi}},
  system~\eqref{eq:closed-loop-system-local}
  is dissipative with respect to
  the supply rate~\eqref{eq:supply-rate},
  if and only if 
  there exists~$P_i  > 0$ such that
    \begin{equation}
      \setlength\arraycolsep{1pt}
      \label{eq:dissipativity-inequality-before-arranging}
      \begin{aligned} 
        &
    \begin{bmatrix}
      I & 0 \\ 
      \hat{A}_i^\top & \hat{C}_i^\top
    \end{bmatrix}^\top 
    \begin{bmatrix}
      P_i & 0 \\
      0 & - P_i
    \end{bmatrix}
    \begin{bmatrix}
      I & 0 \\ 
      \hat{A}_i^\top & \hat{C}_i^\top
    \end{bmatrix}
    \\+ &
    \begin{bmatrix}
      0 & I \\ 
      B_{2,i}^\top & D_{2,i}^\top 
    \end{bmatrix}^\top 
    \begin{bmatrix}
      0 & - I \\
      I & 0 
    \end{bmatrix}
    \begin{bmatrix}
      H_i & G_i^\top \\
      G_i & F_i
    \end{bmatrix}
    \begin{bmatrix}
      0 & - I \\
      I & 0
    \end{bmatrix}
    \begin{bmatrix}
      0 & I \\ 
      B_{2,i}^\top & D_{2,i}^\top 
    \end{bmatrix} 
    \geq 0 .
  \end{aligned}
    \end{equation}
  Rearranging \eqref{eq:dissipativity-inequality-before-arranging} leads to the QMI
  \revised{
  \begin{small}
  \renewcommand{\arraystretch}{1.15} %
  \begin{align} \label{eq:dissipativity-inequality-before-s-lemma}
    \setlength\arraycolsep{3pt}
    [\bullet]^\top 
    \bar{M}_i
    \left[
    \begin{array}{cc}
      I \\ \hdashline 
      \begin{matrix}
        A_i^\top & C_i^\top
      \end{matrix} \\
      \begin{matrix}
        B_{1,i}^\top & D_{1,i}^\top
      \end{matrix} \\
      \begin{matrix}
        B_{2,i}^\top & D_{2,i}^\top
      \end{matrix}
    \end{array}
    \right] \geq 0 ,
  \end{align}
  \renewcommand{\arraystretch}{1}\end{small} %
  where we define $\bar{M_i}$ as 
  \begin{equation*}
    \bar{M_i} \coloneqq 
    \left[
    \begin{array}{cc;{3pt/3pt}ccc}
      P_i & 0 & 0 & 0 & 0 \\
      0 & - F_i & 0 & 0 & G_i \\ \hdashline
      0 & 0 & - P_i & - P_i K_i^\top & 0 \\ 
      0 & 0 & - K_i P_i & - K_i P_i K_i^\top & 0 \\ 
      0 & G_i^\top & 0 & 0 & - H_i
    \end{array}
    \right] .
  \end{equation*}}

  System \eqref{eq:closed-loop-system-local} is dissipative for all matrices $(A_i,B_{1,i},B_{2,i},C_i,D_{1,i},D_{2,i}) \in \Sigma_i$, if and only if \eqref{eq:dissipativity-inequality-before-s-lemma} is satisfied
  for all~$(A_i,B_{1,i},B_{2,i},C_i,D_{1,i},D_{2,i})$ satisfying~\eqref{eq:quadratic-bound-before-s-lemma}.
  By applying Proposition~\ref{prop:matrix-S-lemma},
  this is the case if and only if there exist~$\alpha_i \geq 0$ and~$P_i  > 0$
  such that 
  \begin{equation} \label{eq:after-matrix-s-lemma}
    \begin{aligned}
    \bar{M}_i
    - \alpha_i
    J_i  \geq 0.
  \end{aligned}
  \end{equation}
  Applying the Schur complement \revised{lemma}~\cite{horn2012matrix}
  and considering the change of variables~$L_i \coloneqq K_i P_i$, 
  the feasibility of \eqref{eq:after-matrix-s-lemma} is equivalent to that of \eqref{eq:local-controller-design-lmi}.
  Finally, from the change of variables, the state-feedback gain can be retrieved by~$K_i = L_i P_i^{-1}$ and the storage function can be constructed as~\revised{$V_i(x_i) = x_i^\top P_i^{-1} x_i$},
  which concludes the proof.
\end{proof}

Theorem \ref{thm:local-controller-design} provides a data-driven method to design a dissipative controller.
Given the data $\mathcal{D}_i$ and the matrix $\Phi_i$ from Assumption \ref{assumption:noise-matrices-quadratic-bound}, the matrix $J_i$ can be computed from \eqref{eq:data-matrices-in-the-theorem}.
Then, by constructing matrices $\hat{M}_i$ and $\hat{N}_i$ with the decision variables $P_i,L_i$ and $\alpha_i$, the inequality \eqref{eq:local-controller-design-lmi} can be formed.
By solving \eqref{eq:local-controller-design-lmi}, a state-feedback gain $K_i$ that makes system \eqref{eq:closed-loop-system-local} dissipative can be computed.
Inequality \eqref{eq:local-controller-design-lmi} is an LMI and can be solved efficiently by semidefinite programming methods \cite[Chapter 4.6.2]{boyd2004convex}.

\revised{
  \subsubsection*{Conservatism of Theorem~\ref{thm:local-controller-design}}
  Theorem~\ref{thm:local-controller-design} gives a necessary and sufficient condition for the dissipative state-feedback design problem, when the supply rate matrices satisfy the inertia condition~\eqref{eq:inertia-condition}.
  However, conservatism might arise from a mismatch between the QMI noise characterization in Assumption \ref{assumption:noise-matrices-quadratic-bound} and the actual noise realization.
  In practice, the noise bound should be chosen as tightly as possible based on the prior knowledge of the noise level.
}

\subsubsection*{Implicit conditions on the data matrices}
The feasibility of the LMI \eqref{eq:local-controller-design-lmi} implicitly requires conditions on the data $\mathcal{D}_i$ \revised{including a necessary condition on the data length}.
  For example, inequality \eqref{eq:local-controller-design-lmi} is satisfied only if the inequality
  \begin{equation*}
      \left[\begin{array}{cccccc}
        - P_i & - L_i^\top & 0 \\ 
        - L_i & 0 & 0 \\ 
        0 & 0 & - H_i \\ 
      \end{array}
      \right]
      - 
      \alpha_i 
      \left[
      \begin{array}{c}
        X_i \\
        U_i \\
        V_i
      \end{array}
      \right]
      \Phi_{22, i}\left[
        \begin{array}{c}
          X_i \\
          U_i \\
          V_i
        \end{array}
        \right]^\top
        \geq 0
  \end{equation*}
  holds.
  This is satisfied for some $P_i > 0$ only if $\alpha_i > 0$ and $X_i$ has full row rank, which can be interpreted as an ``excitation'' condition of the data matrix $X_i$.
  \revised{This also imposes a necessary condition on the data length $N_i \geq n_i$.}
  Similar conditions can also be derived for other data matrices.

\subsubsection*{Positive eigenvalue assumption}
We note that the positive eigenvalue assumption on $J_i$ in Theorem \ref{thm:local-controller-design} is not restrictive. 
First, under the additional assumption that the matrix $
\begin{bmatrix}
  X_i^\top & U_i^\top & V_i^\top
\end{bmatrix}^\top$ has full row rank, the positive eigenvalue assumption on $J_i$ is satisfied by an appropriate choice of $\Phi_i$. 
Specifically, for $J_i$ that does not have any positive eigenvalues, we can replace $\Phi_{11,i}$ with $\Phi_{11,i} + \varepsilon I$ with an arbitrarily small $\varepsilon > 0$ and show that the resulting $J_i$ has at least one positive eigenvalue.
This can be interpreted as increasing the noise bound in Assumption \ref{assumption:noise-matrices-quadratic-bound} by an arbitrarily small amount.
Second, even when $J_i$ does not have any positive eigenvalues, the feasibility of the LMI \eqref{eq:local-controller-design-lmi} is sufficient for the design of a dissipativity-inducing controller, which can be deduced by considering only the sufficiency of the matrix S-lemma (Proposition~\ref{prop:matrix-S-lemma}).

\revised{\subsubsection*{Comparison with other works}
Theorem~\ref{thm:local-controller-design} shares technical similarities with other works on data-driven dissipativity based on state-feedback control~\cite{nguyen2024synthesis,kristovic2024data,tanaka2024algebraic}.
However, we assume all system matrices are unknown, whereas other works assume partial knowledge of them.
Moreover, we treat $(F_i, G_i, H_i)$ as decision variables, which enables us to link the results to the decentralized control design method.
This perspective is absent in other works, which consider individual systems only. 
In this sense, Theorem~\ref{thm:local-controller-design} is a necessary component of the overall decentralized synthesis framework rather than a standalone contribution.}

\revisedtwo{
  \subsubsection*{Relation to dissipativity indices}
  Dissipativity indices~\cite{brogliato2007dissipative, zhu2014passivity} are scalar quantities that quantify how dissipative a system is with respect to a given supply rate structure.
  Since the framework treats $(F_i, G_i, H_i)$ as free matrix decision variables, it is in principle more general than any dissipativity index parametrization.
  Moreover, given the optimized $(F_i, G_i, H_i)$, the dissipativity indices consistent with the resulting supply rate and the data can be computed post-hoc by solving another LMI.
  The indices obtained from this procedure are tight, since Theorem~\ref{thm:local-controller-design} is necessary and sufficient with respect to the data.
}

\subsection{Data-driven decentralized stability condition} \label{subsec:data-driven-decen-control-unknown-systems}
Next, we derive a data-driven decentralized stability condition.
\revisedtwo{The} starting point is the following proposition from \cite[Proposition 2]{martinelli2023interconnection}.

\begin{prop} \label{prop:global-stability-condition}
  Assume that the nominal closed-loop dynamics of the $i$-th subsystem \eqref{eq:closed-loop-system-local} with a controller $u_i = K_i x_i$ are dissipative with respect to the supply rate~$s(v_i(t),y_i(t)) = 
  \left[
    \begin{smallmatrix}
      v_i(t) \\ y_i(t)  
    \end{smallmatrix}
  \right]^\top 
\left[
  \begin{smallmatrix}
    H_i & G_i^\top \\
    G_i & F_i
\end{smallmatrix}
\right]^{-1}
\left[
  \begin{smallmatrix}
    v_i(t) \\ y_i(t)  
  \end{smallmatrix}
\right]
$
for all $i \in \{1,\dots, k\}$,
where the matrix $
\left[
  \begin{smallmatrix}
    H_i & G_i^\top \\
    G_i & F_i
  \end{smallmatrix}
\right]
$ is invertible.
Define $F \coloneqq \operatorname*{diag}(F_1,\dots,F_k)$, 
  $G \coloneqq \operatorname*{diag}(G_1,\dots,G_k)$, and 
  $H \coloneqq \operatorname*{diag}(H_1,\dots,H_k)$.
The origin of the nominal global closed-loop system \eqref{eq:global-dynamics} is asymptotically stable if
\begin{subequations} \label{eq:global-stability-condition}
  \begin{eqnarray}
    MFM^\top - MG - G^\top M^\top + H > 0, \label{eq:global-stability-condition-M} \\ 
    F \leq 0.\label{eq:global-stability-condition-F} 
  \end{eqnarray}
\end{subequations}
\end{prop}

Condition \eqref{eq:global-stability-condition} in Proposition \ref{prop:global-stability-condition} can be used as a sufficient condition to verify whether the nominal global closed-loop system \eqref{eq:global-dynamics} is asymptotically stable, given the supply rate of each subsystem.
For example, if all systems are passive, i.e., if $(F_i,G_i,H_i) = (0, \frac{1}{2}I, 0)$ for all $i \in \{1,\dots,k\}$, condition \eqref{eq:global-stability-condition} reduces to $M+M^\top < 0$, which is a well-known stability condition for the interconnection of passive systems.

Condition \eqref{eq:global-stability-condition} requires the full information of the interconnection matrix $M$ as well as the supply rate matrices from all subsystems, making it unsuitable for data-driven and decentralized control. Hence, we first derive a sufficient condition characterized in terms of local information of the interconnection, in the form of the matrix $\tilde{M}_i^\mathrm{r}$.

\begin{lemma} \label{lemma:decentralized-known}
  Define the matrices $E_i$ and $\Lambda_i$ as
  \begin{subequations} \label{eq:construction-lambda_i}
    \begin{align}
    E_i &\coloneqq 
  \begin{bmatrix}
     I_{p_i} & 0_{p_i \times p_{\sigma(1)}} & \cdots & 0_{p_i \times p_{\sigma(|\pazocal{N}_i|-1)}}
  \end{bmatrix}, \label{eq:E_i} \\
    \Lambda_i &\coloneqq  
  \begin{bmatrix} 
    \bullet
  \end{bmatrix}^\top \begin{bmatrix} H_i -\beta_i I & G_i^\top \\ G_i & F_i \end{bmatrix} \begin{bmatrix} E_i &0 \\ 0 & I_{p_i} \end{bmatrix} . \label{eq:lambda_i}
  \end{align}
  \end{subequations}
  Condition \eqref{eq:global-stability-condition} is satisfied if,
  for all $i\in\{1,\dots,k\}$,
  there exists $\beta_i > 0$ such that
  \begin{subequations} \label{eq:decentralized-lmi-condition-stability}
    \setlength\arraycolsep{1pt}
        \begin{align}
          \begin{bmatrix}
            I_{\tilde{p}_i} \\ 
            \tilde{M}_i^\mathrm{r}
          \end{bmatrix}^\top 
          \Lambda_i
          \begin{bmatrix}
            I_{\tilde{p}_i} \\
            \tilde{M}_i^\mathrm{r}
          \end{bmatrix} 
          &\geq 0 , \label{eq:decentralized-stab-cond-with-Mr} \\
          F_i &\leq 0 . \label{eq:F_i-decentralized-condition}
        \end{align}
      \end{subequations}
\end{lemma}

\begin{proof}
First, \eqref{eq:global-stability-condition-F} is satisfied if and only if for each $i\in\{1,\dots,k\}$, the inequality \eqref{eq:F_i-decentralized-condition} is satisfied.
Next, we define 
\begin{equation*}
  M^\mathrm{r}_i 
    \coloneqq 
  \begin{bmatrix}
    M_{i1} & \cdots & M_{ik}
  \end{bmatrix} 
  \in \mathbb{R}^{p_i \times p}.
\end{equation*}
Then, from $M= M^\top$, \eqref{eq:global-stability-condition-M} can be written as 
  \begin{equation} \label{eq:global-stability-condition-decomposed}
    \begin{aligned}
      \left( \sum_{i = 1}^k (M^\mathrm{r}_i)^\top F_i M^\mathrm{r}_i\right)
    - \begin{bmatrix}
      (M^\mathrm{r}_1)^\top G_1 & \cdots & (M^\mathrm{r}_k)^\top G_k
    \end{bmatrix} \\
    - \begin{bmatrix}
      G_1^\top M^\mathrm{r}_1  \\
      \vdots \\
      G_k^\top M^\mathrm{r}_k
    \end{bmatrix}
    + \operatorname*{diag}(H_1,\dots,H_k) > 0 .
    \end{aligned}
  \end{equation}
  By considering the positive definiteness of the $i$-th component,
  inequality \eqref{eq:global-stability-condition-decomposed} is satisfied
  if for each $i\in\{1,\dots,k\}$, there exists $\beta_i>0$ such that
  \begin{equation} \label{eq:M_i-decentralized-condition}
    \begin{split}
    (M^\mathrm{r}_i)^\top F_i M^\mathrm{r}_i
  -\begin{bmatrix} 
    0 & \cdots & (M^\mathrm{r}_i)^\top G_i & \cdots & 0
  \end{bmatrix} \\
  -
  \begin{bmatrix} 
    0 \\ 
    \vdots \\
    G_i^\top M^\mathrm{r}_i \\
    \vdots \\
    0
  \end{bmatrix}
  + \operatorname*{diag}(0,\dots,H_i - \beta_i I,\dots,0)
  \geq  0
  \end{split} 
  \end{equation}
  is satisfied.
  Since $M_{ij} = 0$ holds for all $j \notin \pazocal{N}_i $, the size of \eqref{eq:M_i-decentralized-condition} can be reduced and it is equivalent to 
  \begin{equation*}
    \begin{bmatrix} E_i^\top & (\tilde{M}^\mathrm{r}_i)^\top \end{bmatrix} \begin{bmatrix} H_i -\beta_i I & G_i^\top \\ G_i & F_i \end{bmatrix} \begin{bmatrix} E_i \\ \tilde{M}^\mathrm{r}_i \end{bmatrix} \geq 0
  \end{equation*}
  and thus to \eqref{eq:decentralized-stab-cond-with-Mr}.
\end{proof}

We next propose a data-driven decentralized stability condition for the global system.
From the assumptions on the measurements~\eqref{eq:interconnection-output-local-relation}, the matrices satisfy 
\begin{equation} \label{eq:data-relation-int}
  \tilde{V}_i = \tilde{M}^\mathrm{r}_i \tilde{Y}_i + \Xi_i .
\end{equation}
Then, the set of matrices $\tilde{M}_i^\mathrm{r}$ that are {consistent with the data} is defined as
\begin{equation*}
  \Sigma_{\tilde{M}^\mathrm{r}_i }
  \coloneqq
  \left\{ \tilde{M}^\mathrm{r}_i  :\text{$\exists \Xi_i $ s.t.~\eqref{eq:interconnection-matrix-relation},~\eqref{eq:data-relation-int} hold} \right\}.
\end{equation*}
Condition \eqref{eq:decentralized-stab-cond-with-Mr} cannot be directly used since $\tilde{M}_i^\mathrm{r}$ is unknown.
Hence, we propose a data-driven condition which is feasible if and only if the stability condition \eqref{eq:decentralized-stab-cond-with-Mr} is feasible for all matrices $\tilde{M}_i^\mathrm{r}$ consistent with data.

\begin{theorem} \label{thm:decentralized-unknown}
  \revised{Consider the data $\tilde{\mathcal{D}}_i = \{\tilde{V}_i,\tilde{Y}_i\}$ and assume that it satisfies Assumption~\ref{assumption:interconnection-noise-matrices-QMI-local}.}
  Define the matrix $\Theta_i$ as 
    \begin{equation} \label{eq:def-xi-i}
      \Theta_i \coloneqq  
    \begin{bmatrix} I_{p_i} & \tilde{V}_i \\ 
      0 & - \tilde{Y}_i
      \end{bmatrix}
      \Psi_i
    \begin{bmatrix} I_{p_i} & \tilde{V}_i \\ 
      0 & - \tilde{Y}_i
      \end{bmatrix}^\top,
  \end{equation}
  and assume that it is invertible.
  Define the matrix $\hat{\Theta}_i$ as 
  \begin{equation} \label{eq:def-xi-hat-i}
    \hat{\Theta}_i
    \coloneqq 
    \begin{bmatrix}
      0 & - I_{\tilde{p}_i} \\
      I_{p_i} & 0
    \end{bmatrix}
    \Theta_i^{-1}
    \begin{bmatrix}
      0 & - I_{p_i} \\
      I_{\tilde{p}_i} & 0
    \end{bmatrix}.
  \end{equation}
  Then, the stability condition \eqref{eq:decentralized-stab-cond-with-Mr} holds for all $\tilde{M}^\mathrm{r}_i \in \Sigma_{\tilde{M}^\mathrm{r}_i }$,
  if and only if there exists $\tau_i \geq 0$ such that 
  \begin{equation} \label{eq:local-lmi-condition-for-data-driven-stability}
    \Lambda_i - \tau_i \hat{\Theta}_i\geq 0 .
  \end{equation}
\end{theorem}

\begin{proof}
  The data matrices $\tilde{V}_i,\tilde{Y}_i$ and the noise matrix $\Xi_i$ satisfy
  \begin{equation} \label{eq:matrix-M_i-and-W_3-relation}
    \begin{bmatrix} I_{p_i} \\ (\Xi_i)^\top \end{bmatrix} 
    = \begin{bmatrix} I_{p_i} & \tilde{V}_i \\ 
      0 & - \tilde{Y}_i
      \end{bmatrix}^\top
      \begin{bmatrix} 
        I_{p_i} \\ (\tilde{M}^\mathrm{r}_i)^\top 
      \end{bmatrix} .
  \end{equation}
  Hence, by substituting \eqref{eq:matrix-M_i-and-W_3-relation} into \eqref{eq:interconnection-matrix-relation},
  $\tilde{M}^\mathrm{r}_i \in \Sigma_{\tilde{M}^\mathrm{r}_i}$ holds if and only if the
  matrix $\tilde{M}^\mathrm{r}_i$ satisfies 
  \begin{equation} \label{eq:interconnection-matrix-M_i-qmi}
    \begin{bmatrix}
      I_{p_i} \\ (\tilde{M}^\mathrm{r}_i)^\top
    \end{bmatrix}^\top 
    \Theta_i
    \begin{bmatrix}
      I_{p_i} \\ (\tilde{M}^\mathrm{r}_i)^\top
    \end{bmatrix} 
    \geq 0 .
  \end{equation}
  From Proposition \ref{prop:dualization-lemma},
  the matrix $\tilde{M}^\mathrm{r}_i$ satisfies \eqref{eq:interconnection-matrix-M_i-qmi} if and only if 
  \begin{equation} \label{eq:M_i^c-before-s-lemma}
    \begin{bmatrix}
      I_{\tilde{p}_i} \\
      \tilde{M}^\mathrm{r}_i
    \end{bmatrix}^\top 
    \hat{\Theta}_i
    \begin{bmatrix}
      I_{\tilde{p}_i} \\
      \tilde{M}^\mathrm{r}_i
    \end{bmatrix} 
    \geq 0 .
  \end{equation}

  By construction, it can be shown that $\Theta_i \in \mathbf{\Pi}_{p_i, \tilde{p}_i}$ holds.
  Additionally, with the assumption that $\Theta_i$ is invertible, $\operatorname{In}(\Theta_i) = (\tilde{p}_i, 0, p_i)$ holds \cite[Theorem 3.2]{van2023quadratic}.
  Hence, from \eqref{eq:def-xi-i}, \eqref{eq:def-xi-hat-i} and the Sylvester's law of inertia \cite[Theorem 4.5.8]{horn2012matrix}, it holds that $\operatorname{In}(\hat{\Theta}_i) = (p_i, 0, \tilde{p}_i)$, implying $\hat{\Theta}_i \in \mathbf{\Pi}_{\tilde{p}_i,p_i}$ and that $\hat{\Theta}_i$ has at least one positive eigenvalue.
  Hence, from the matrix S-lemma (Proposition \ref{prop:matrix-S-lemma}), 
  the stability condition \eqref{eq:decentralized-stab-cond-with-Mr} is satisfied for all $\tilde{M}_i^\mathrm{r}$ such that 
  \eqref{eq:M_i^c-before-s-lemma} is satisfied, if and only if there exists $\tau_i \geq 0$ such that \eqref{eq:local-lmi-condition-for-data-driven-stability} is satisfied.
\end{proof}

Theorem \ref{thm:decentralized-unknown} provides a data-driven decentralized LMI condition that ensures the stability of the global system \eqref{eq:global-dynamics}. This can be used as follows: assume that the $i$-th closed-loop system \eqref{eq:closed-loop-system-local} is dissipative with respect to the supply rate $s_i(v_i(t),y_i(t)) = \left[\begin{smallmatrix}
  v_i(t) \\ y_i(t)  
\end{smallmatrix}\right]^\top 
\left[
\begin{smallmatrix}
      H_i & G_i^\top \\
      G_i & F_i
    \end{smallmatrix}\right]^{-1}
    \left[
      \begin{smallmatrix}
        v_i(t) \\ y_i(t)  
      \end{smallmatrix}
    \right]$.
The matrices $\Theta_i$ and $\hat{\Theta}_i$ can be constructed from \eqref{eq:def-xi-i} and \eqref{eq:def-xi-hat-i}, given the data $\tilde{\mathcal{D}}_i = \{\tilde{V}_i, \tilde{Y}_i\}$ and the noise matrix $\Psi_i$ in Assumption \ref{assumption:interconnection-noise-matrices-QMI-local}.
The matrix $\Lambda_i$ can also be constructed from \eqref{eq:construction-lambda_i}, given the supply rate matrices $(F_i,G_i,H_i)$ of the subsystem.
Then, if the condition \eqref{eq:local-lmi-condition-for-data-driven-stability} is satisfied for all $i \in \{1,\dots,k\}$, then for all matrices $\tilde{M}_i^\mathrm{r}$ consistent with the data, condition \eqref{eq:decentralized-stab-cond-with-Mr} holds, which certifies the stability of the global closed-loop system \eqref{eq:global-dynamics}.

\subsubsection*{Invertibility assumption on $\Theta_i$}
Similar to the positive eigenvalue assumption in Theorem \ref{thm:local-controller-design}, we can show that the invertibility assumption on $\Theta_i$ in Theorem \ref{thm:decentralized-unknown} is not overly restrictive.
Under the additional assumption such that the matrix $\tilde{Y}_i$ has full row rank, the invertibility assumption on $\Theta_i$ in Theorem~\ref{thm:decentralized-unknown} is satisfied by an appropriate choice of $\Psi_i$.
Namely, even if $\Theta_i$ is not invertible, we can replace $\Psi_{11,i}$ with $\Psi_{11,i} + \varepsilon I$ with an arbitrarily small $\varepsilon > 0$ and show that the resulting $\Theta_i$ is invertible.
As before, this can be interpreted as increasing the noise bound in Assumption \ref{assumption:interconnection-noise-matrices-QMI-local} with an arbitrarily small amount.
\revised{The invertibility assumption on $\Theta_i$ also imposes a necessary condition on the data length $\tilde{N}_i$. 
For example, $\tilde{N}_i \geq \tilde{p}_i$ is required for the matrix $\tilde{Y}_i$ to have full row rank, which is a sufficient condition for the invertibility of $\Theta_i$ to be guaranteed by an appropriate choice of $\Psi_i$.}

\subsection{Control algorithm for Problem \ref{prb:decentralized-unknown}} \label{subsec:data-driven-decentralized-control-alg}
By combining the LMIs obtained in Sections~\ref{sec:local_controller_design} and \ref{subsec:data-driven-decen-control-unknown-systems}, we propose a data-driven decentralized control design algorithm for Problem~\ref{prb:decentralized-unknown}.
Each subsystem can be made dissipative with respect to a quadratic supply rate \eqref{eq:supply-rate} by solving the LMI \eqref{eq:local-controller-design-lmi} in Theorem~\ref{thm:local-controller-design}. Moreover, the LMI \eqref{eq:local-lmi-condition-for-data-driven-stability} in Theorem~\ref{thm:decentralized-unknown} gives a sufficient decentralized condition for the asymptotic stability of the global closed-loop system, given the supply rate matrices $(F_i,G_i,H_i)$. Here, the matrices $(F_i,G_i,H_i)$ appear linearly in both \eqref{eq:local-controller-design-lmi} and \eqref{eq:local-lmi-condition-for-data-driven-stability}, which allows us to treat them as decision variables.
\revised{By solving  \eqref{eq:local-controller-design-lmi}, \eqref{eq:F_i-decentralized-condition} and \eqref{eq:local-lmi-condition-for-data-driven-stability} with respect to the original decision variables $ P_i, L_i, \alpha_i, \beta_i$ and $\tau_i$ and additionally with the supply rate matrices $(F_i,G_i,H_i)$ that satisfy the inertia condition \eqref{eq:inertia-condition}, local controllers can be designed such that the asymptotic stability of the \revisedtwo{nominal} global system is satisfied. 
}
The procedure is summarized in Algorithm~\ref{alg:decentralized-unknown}.

\begin{algorithm}[h]
  \caption{Data-driven decentralized control}
  \begin{algorithmic}[1] \label{alg:decentralized-unknown}
  \FOR {$i\in\{1,\dots,k\}$}
  \STATE \textbf{Input:} Data $ \mathcal{D}_i $, $\tilde{\mathcal{D}}_{i}$
  \STATE Find $(F_i,G_i,H_i)$, $P_i>0$, $L_i$, $\alpha_i \geq 0$, {$\beta_i > 0$} and $\tau_i \geq 0$
  such that 
  the inertia condition \eqref{eq:inertia-condition} and the LMIs \eqref{eq:local-controller-design-lmi}, \eqref{eq:F_i-decentralized-condition}, \eqref{eq:local-lmi-condition-for-data-driven-stability} are satisfied\label{state:decentralized-control-unknown-lmi}
  \STATE Compute $K_i = L_i P_i^{-1}$
  \ENDFOR
  \RETURN $\{K_i\}_{i=1}^k$
  \end{algorithmic}
\end{algorithm}

\revised{
  The inertia condition \eqref{eq:inertia-condition} is a nonconvex constraint that makes the feasibility problem in step~\ref{state:decentralized-control-unknown-lmi} nonconvex.
  However, one can exploit additional structure to guarantee that \eqref{eq:inertia-condition} is satisfied by construction while yielding a convex problem. 
  For example, for $p_i = 1$, adding the convex constraints $H_i > 0$ and revising \eqref{eq:F_i-decentralized-condition} to $F_i < 0$ for step~\ref{state:decentralized-control-unknown-lmi} ensures that the inertia condition is satisfied.
  Moreover, \eqref{eq:local-lmi-condition-for-data-driven-stability} is an LMI whose size depends only on the number of neighbors, and requires only the output measurements of these neighbors, enabling decentralized control design.
}

\section{Diffusive coupling} 
\label{sec:control-design-for-unknown-diffusive-coupling}
In this section, we consider Problem \ref{prb:diffusive-coupling}, where the interconnection structure is diffusive.
We derive a condition that guarantees the stability of the global system and propose a method to compute the tightest bound consistent with the data.
Combining these results with the LMI \eqref{eq:local-controller-design-lmi} leads to a data-driven decentralized control algorithm.

\subsection{Stability condition based on the upper bound of $d_i$} \label{subsec:data-driven-stability-diffusive}
Recall that the diffusive coupling \eqref{eq:diffusive-coupling} is characterized by scalars $\{a_{ij}\}_{j \in \pazocal{N}_i}$ and the weighted degree is defined as $d_i = \sum_{j \in \pazocal{N}_i} a_{ij}$.
The following proposition from~\cite[Corollary 1]{martinelli2023interconnection} provides a sufficient condition
on the matrices $(F_i, G_i, H_i)$ and the weighted degree $d_i$,
such that the nominal global closed-loop system \eqref{eq:global-dynamics} is asymptotically stable.

\begin{prop}  \label{prop:decentralized-stability-conditions}
  Assume that for all $i \in \{1,\dots,k\}$, the interconnection is the diffusive coupling \eqref{eq:diffusive-coupling} and the nominal closed-loop dynamics of the $i$-th subsystem \eqref{eq:closed-loop-system-local} are dissipative with respect to the supply rate~$s_i(v_i(t),y_i(t))=
  \left[
    \begin{smallmatrix}
        v_i(t) \\ y_i(t) 
    \end{smallmatrix}\right]^\top 
    \left[
    \begin{smallmatrix}
      H_i & G_i^\top \\
      G_i & F_i
    \end{smallmatrix}
    \right]^{\revised{-1}}
    \left[
    \begin{smallmatrix}
      v_i(t) \\ y_i(t)  
    \end{smallmatrix}
    \right]$.
  Then, the nominal global closed-loop system \eqref{eq:global-dynamics} is asymptotically stable if the following condition is satisfied for all $i \in \{1,\dots, k\}$:
  \begin{equation}
    G_i = \frac{1}{2} \alpha I,~ - \frac{1}{2 d_i}I < F_i < 0,~H_i > 2d_i \tilde{\alpha} I,  \label{eq:diffusive-coupling-decentralized-condition}
  \end{equation}
  where $\alpha \in \mathbb{R}$ and $\tilde{\alpha} \coloneqq \max\left\{ 1- \alpha, 0\right\}$ are given scalars.
\end{prop}

In Problem \ref{prb:diffusive-coupling}, the weighted degree $d_i$ is unknown since $a_{ij}$ is unknown, and Proposition \ref{prop:decentralized-stability-conditions} cannot be directly used to ensure the stability of the global system.
Instead, we can show that it suffices to ensure the same condition for some $d_i' \geq d_i$.
Namely, if 
\begin{equation} \label{eq:diffusive-coupling-decentralized-condition-approximated}
  G_i = \frac{1}{2} \alpha I,~- \frac{1}{2 d_i'}I < F_i < 0,~H_i > 2d_i' \tilde{\alpha} I
\end{equation}
is satisfied for some $d_i' \geq d_i$, then \eqref{eq:diffusive-coupling-decentralized-condition} holds, hence the stability of the global system is guaranteed.

\subsection{Data-driven upper bound on $d_i$}
To use \eqref{eq:diffusive-coupling-decentralized-condition-approximated} as a stability condition in \revisedtwo{the} control algorithm, we wish to compute the tightest upper bound of the weighted degree $d_i$, given the interconnection data $\tilde{\mathcal{D}}_i = \{\tilde{V}_i,\tilde{Y}_i\}$. 
To this end, we introduce the set of weighted degrees $d_i$ {consistent with the data} 
\begin{equation*}
  \begin{aligned}
    \Sigma_{d_i}
  \coloneqq &
  \left\{ d_i = \sum_{j=1}^k a_{ij} \mid
  \exists  \tilde{W}_i,\{a_{ij}\}_{j=1}^k ~\text{such that} \right. \\
  &
  \left.
  \text{\eqref{eq:diffusive-coupling} for all $t \in \{0,\dots,\tilde{N}_i-1\}$ and \eqref{eq:interconnection-matrix-relation} hold}\right\},
  \end{aligned}
\end{equation*}
and define the maximum value of $d_i \in \Sigma_{d_i}$ as $d_i^{\max}$.
We show in the following theorem that $d_i^{\max}$ can be computed as a solution of a quadratic program.

\begin{theorem} \label{thm:closed-form-degree}
  \revisedtwo{Assume that $p_i=1$ for all $i \in \{1,\dots, k \}$.
  Consider the data $\tilde{\mathcal{D}}_i = \{\tilde{V}_i,\tilde{Y}_i\}$ and assume that it satisfies Assumption~\ref{assumption:interconnection-noise-matrices-QMI-local}.}
  Define $\revised{\Theta_i} \in \mathbb{R}^{(|\pazocal{N}_i|+1) \times (|\pazocal{N}_i|+1)}$ as in \eqref{eq:def-xi-i}, assume that it is invertible, and let $\hat{\Theta}_i \in \mathbb{R}^{(|\pazocal{N}_i|+1) \times (|\pazocal{N}_i|+1)}$ as in \eqref{eq:def-xi-hat-i}.
  Define
  \begin{equation*}
    \Upsilon_i = 
  \begin{bmatrix}
    - e_1^\top & 0 \\
    0 & 1
  \end{bmatrix}
  \hat{\Theta}_i
  \begin{bmatrix}
    - e_1 & 0 \\
    0 & 1
  \end{bmatrix} \in \mathbb{R}^{2 \times 2} ,
  \end{equation*}
  where $e_1 = \begin{bmatrix}
    1 & 0 & \cdots & 0
  \end{bmatrix}^\top \in \mathbb{R}^{|\pazocal{N}_i|}$.
  Then,
  \begin{equation} \label{eq:degree-quadratic-bound}
    d_i^\mathrm{max} = \max d_i 
    ~ \text{such that}~
    \begin{bmatrix}
      1 \\
      d_i
    \end{bmatrix}^\top 
    \Upsilon_i
    \begin{bmatrix}
      1 \\
      d_i
    \end{bmatrix}
    \geq 0 .
  \end{equation}
\end{theorem}

\begin{proof}
  Since $d_i = - M_{ii}$, it holds that $d_i = - \tilde{M}_i^\mathrm{r} e_1$, and we have
  \begin{equation} \label{eq:degree-set-data-consistent}
    \begin{aligned}
      \Sigma_{d_i}
    \coloneqq &
    \left\{ - \tilde{M}_i^\mathrm{r} e_1 \mid
    \exists  \tilde{W}_i 
    \text{~such that \eqref{eq:interconnection-matrix-relation} and \eqref{eq:data-relation-int} hold}\right\} .
    \end{aligned}
  \end{equation}
  From the proof of Theorem \ref{thm:decentralized-unknown}, $\tilde{M}_i^{\mathrm{r}}$ is characterized as \eqref{eq:M_i^c-before-s-lemma}, and since $
    d_i = - \tilde{M}_i^\mathrm{r} e_1$,
  we can apply \cite[Theorem 3.4]{van2023quadratic} and obtain the quadratic program \eqref{eq:degree-quadratic-bound}.
\end{proof}

\subsection{Control algorithm for Problem \ref{prb:diffusive-coupling}}
By combining the results developed in Section \ref{subsec:data-driven-stability-diffusive} with the local control design LMI \eqref{eq:local-controller-design-lmi}, we propose a data-driven decentralized control algorithm for Problem \ref{prb:diffusive-coupling}.
We first compute the maximum weighted degree consistent with data $d_i^{\max}$ using \eqref{eq:degree-quadratic-bound}.
Then, we combine the LMIs \eqref{eq:local-controller-design-lmi} and \eqref{eq:diffusive-coupling-decentralized-condition-approximated} with $d_i' = d_i^{\max}$, and treat the supply rate matrices $(F_i,G_i,H_i)$ as decision variables along with $P_i$, $L_i$ and $\alpha_i$,
as shown in Algorithm \ref{alg:decentralized-diffusive-coupling}.
Note that the diffusive coupling condition has to be given to the subsystems beforehand and the parameter $\alpha$ in \eqref{eq:diffusive-coupling-decentralized-condition-approximated} has to be agreed for all subsystems.
\revised{We also note that for $p_i = 1$, the inertia 
condition~\eqref{eq:inertia-condition} is satisfied by 
construction for any $\alpha$, since 
condition~\eqref{eq:diffusive-coupling-decentralized-condition-approximated} 
already enforces $F_i < 0$ and $H_i > 0$. 
The inertia condition is therefore omitted from 
Algorithm~\ref{alg:decentralized-diffusive-coupling}.}

\begin{algorithm}
  \caption{Data-driven decentralized control for diffusive coupling}
  \begin{algorithmic}[1] \label{alg:decentralized-diffusive-coupling}
  \FOR {$i \in  \{1,\dots,k\}$}
  \STATE \textbf{Input:} Data $ \mathcal{D}_i $, $\tilde{\mathcal{D}}_i$
  \STATE Compute $d_i^{\max}$ from \eqref{eq:degree-quadratic-bound} and set $d_i' = d_i^{\max}$
  \STATE Find $(F_i,G_i,H_i)$, $P_i>0$, $L_i$ and $\alpha_i \geq 0$ such that 
  \revised{the LMIs
  \eqref{eq:local-controller-design-lmi}, \eqref{eq:diffusive-coupling-decentralized-condition-approximated} are satisfied} \label{state:diffusive-control-unknown-lmi}
  \STATE Compute $K_i = L_i P_i^{-1}$
  \ENDFOR
  \RETURN $\{K_i\}_{i=1}^k$
  \end{algorithmic}
\end{algorithm}

\section{Interconnected Microgrid Examples}
\label{sec:numerical_examples}
In this section, we apply the proposed control algorithms to direct current (DC) microgrid systems comprising distributed generation units (DGUs), interconnected via transmission lines\footnote{\revisedtwo{Code is available at \href{https://gitlab.ethz.ch/nakanot/ddd-interconnectedsystems-codes}{\texttt{https://gitlab.ethz.ch/nakanot/ddd-interconnectedsystems-codes}}.}}.
We consider the model of DGUs in Figure \ref{Fig:Buck}, where the input $V_{\mathrm{in}, i}$ represents a DC source and the averaged model of a Buck converter is modeled as an RLC circuit with parameters $R,L$ and $C$.
For simplicity, we consider the input voltage $V_{\mathrm{in}, i}$ directly as an input and do not consider acting on the duty cycle of the switch controlling it.
The internal current and the external current from the neighboring DGUs are denoted by~$I_i$ and $I_{G,i}$, respectively.
We also assume constant impedance loads $I_{L,i} = Y V_i$ with the internal voltage denoted by $V_i$ and conductance $Y>0$.

The nominal {discrete-time} dynamics of the $i$-th DGU obtained by zero-order hold discretization with sampling time $T_s$ become
\begin{subequations}\label{eq:dgu-dynamics-discrete-time}
	\begin{align}
		x_i(t+1) & = A_i x_i(t) + B_{1,i} u_i(t) + B_{2,i} v_i(t) \\
		y_i(t) & = C_i x_i(t),
	\end{align}
\end{subequations}
where $x_i = \begin{bmatrix} V_i & I_i \end{bmatrix}^\top$ is the state, $u_i = V_{\mathrm{in},i}$ is the control input, $v_i=I_{G,i}$ is the interconnection input and the system matrices are given by
\begin{align*}
	A_i & = 
	\begin{bmatrix} 
		1-T_sY/C & T_s/C \\
		-T_s/L & \revised{1-T_sR/L}  
	\end{bmatrix} ,~ 
	B_{1,i} = 
	\begin{bmatrix} 
		0 \\ 
		T_s/\revised{L}
	\end{bmatrix}, \\
	B_{2,i} & = \begin{bmatrix} T_s/C \\ 0 \end{bmatrix} ,~
	C_i = 
	\begin{bmatrix} 
		1 & 0 
	\end{bmatrix}.
\end{align*} 

The DGUs are coupled with each other via resistive transmission lines, giving rise to a diffusive coupling 
\begin{equation} \label{interconnectionlines}
	v_i(t) = \sum_{j\in\pazocal{N}_i} \frac{1}{R_{ij}} (y_j(t) - y_i(t)),
\end{equation}
where $R_{ij}>0$ denotes the resistance of the line connecting the $i$-th DGU with the $j$-th DGU,
and $\pazocal{N}_i$ contains the indices of the DGUs connected to the $i$-th DGU.
We consider $k = 50$ DGUs coupled with a ring interconnection plus additional $20$ edges, with the interconnection topology in Fig. \ref{fig:network}.
\revised{The global dynamics do not have an algebraic loop, since $D_{2,i} = 0$ for all $i$.}

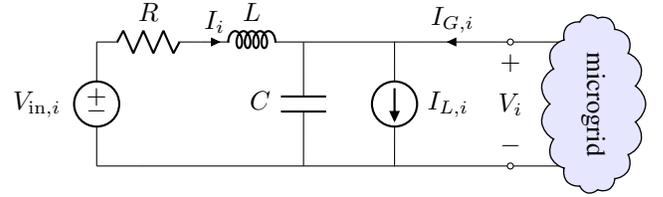
\begin{figure}
	\centering
	\begin{circuitikz}[scale=0.40]
		\ctikzset{bipoles/length=1cm}
		\draw (1,0) to[V, v=$V_{\mathrm{in},i}$, invert] (1,3);
		\draw (1,3) to[R,l=$R$] (3.5,3);
		\draw (3.5,3) to[L,i>^=$I_i$,l=$L$,] (6,3);
		\draw (6,0) to[C,l=$C$] (6,3);
		\draw (6,0) -- (2,0);
		\draw (2,0) -- (1,0); 
		\draw (6,3) -- (8.2,3) to[american current source,l=$I_{L, i}$] (8.2,0) -- (6,0);
		\draw (8.2,3) to[short,i<=$I_{G, i}$] (11,3) to[short,-o] (11,3) -- (12.2,3);
		\draw (8.2,0) to[short,-o] (11,0) -- (12.2,0);
		\node at (11,2.5) {$+$};
		\node at (11,0.5) {$-$};
		\node at (11,1.5) {$V_i$};
		\node [draw, cloud, cloud puffs=15, aspect=2, cloud puff arc=120, line width=0.5, rotate=270, fill=blue!10, minimum width=2cm, minimum height=1cm] at (13,1.5) {microgrid};
	\end{circuitikz}
	\caption{Model of the $i$-th distributed generation unit (DGU).}
	\label{Fig:Buck}
\end{figure}

\begin{figure}
	\centering
	\includegraphics[width=0.6\columnwidth]{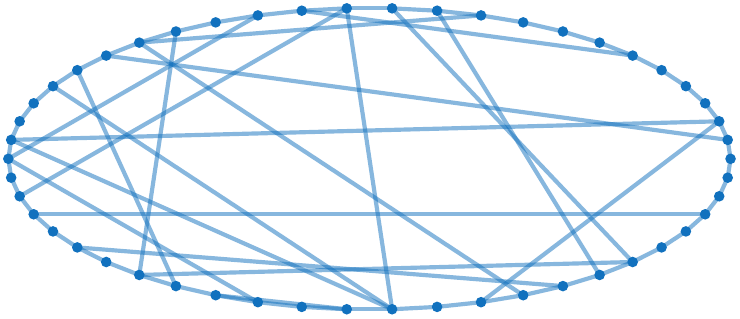}
	\caption{Graph with $k=50$ nodes representing the microgrid model. 
	The nodes represent the DGUs and the edges represent line resistances between two interconnected DGUs.}
	\label{fig:network}
\end{figure}

We consider the situation where the electrical parameters of each DGU as well as the resistance of each transmission line are {unknown}.
Instead, we assume that we can collect $N_i$-length data of the voltage and the current
from the system with noise \eqref{eq:dynamics-of-the-ith-subsystem-with-noise},
where the matrix $W_i$ satisfies Assumption \ref{assumption:noise-matrices-quadratic-bound} with 
$(\Phi_{11, i},\Phi_{12, i},\Phi_{22, i}) = ( N_i \varepsilon_l I, 0, - I)$.
From Section \ref{sec:problem_formulation}, this implies that~$
\left\| \left[\begin{smallmatrix}
	w_{1,i}(t) \\ w_{2,i}(t)
\end{smallmatrix} \right]\right\|_2^2  \leq \varepsilon_l$ holds for all $t = 0,\dots,N_i - 1$.
Since the line resistance $R_{ij}$ is unknown, the interconnection structure \eqref{interconnectionlines} is also unknown.
We assume that we collect~$\tilde{N}_i$-length trajectories from the interconnection relation \eqref{eq:interconnection-output-local-relation}, where the matrix $\tilde{W}_i$ satisfies Assumption \ref{assumption:interconnection-noise-matrices-QMI-local} with 
$(\Psi_{11,i},\Psi_{12,i},\Psi_{22,i}) = (\tilde{N}_i \varepsilon_g I, 0, - I)$.
Again, this implies that the noise vector satisfies $
\left\| \tilde{w}_i(t) \right\|_2^2  \leq \varepsilon_g$ for all $t = 0,\dots,\tilde{N}_i-1$.

We aim to stabilize the internal current $I_i$ and the internal voltage $V_i$ of all DGUs with state-feedback controllers designed in a decentralized fashion on the basis of data.
This problem can be considered as an instance of Problem~\ref{prb:decentralized-unknown} if we assume that the interconnection structure is unknown,
and that of Problem~\ref{prb:diffusive-coupling} if we exploit the fact that the system is coupled by diffusive coupling.
We consider stabilization to the origin for simplicity, but stabilization to any desired steady state can be achieved by the appropriate changes of variables.

Based on the collected data from each DGU, we apply both Algorithms \ref{alg:decentralized-unknown} and \ref{alg:decentralized-diffusive-coupling},
and compute a state-feedback controller for each DGU to make the closed-loop global system stable. \revised{The discrete-time dynamics~\eqref{eq:dgu-dynamics-discrete-time} are 
obtained by zero-order hold discretization with sampling 
time $T_s = 10^{-3}$\,s.}
The electrical parameters used in the experiment are drawn from uniform distributions according to Table \ref{Tab:electricalparameters}.
The lengths of the local data and the interconnection data are both $N_i = \tilde{N}_i = 50$, collected using piecewise constant inputs uniformly drawn from $[-1,1]$ that change once every 5 time steps.
The parameters of the noise matrices are set to $\varepsilon_l = \varepsilon_g = 0.001$.
In Algorithm \ref{alg:decentralized-diffusive-coupling}, we set the parameter $\alpha = 1$.
In both algorithms, we solve the LMIs without the inertia condition~\eqref{eq:inertia-condition} initially, then we check that the matrices~$(F_i,G_i,H_i)$ satisfy the condition.
We solve the LMIs using
YALMIP \cite{lofberg2004yalmip} with MOSEK \cite{mosek} using MATLAB (2025a) on a Macbook Pro with Apple M2 processor and 16 GB of RAM.
The behaviors of the internal currents and voltages of the global closed-loop system without noise are shown in Fig.~\ref{fig:numerical-experiments-alg4} for Algorithm \ref{alg:decentralized-unknown} and Fig.~\ref{fig:numerical-experiments-alg5} for Algorithm \ref{alg:decentralized-diffusive-coupling}.
In both algorithms, the current and the voltage asymptotically converge to the origin, which shows the effectiveness of the controllers. \revised{In addition, Fig.~\ref{fig:iss} shows the closed-loop trajectories under nonzero noise satisfying Assumptions~\ref{assumption:noise-matrices-quadratic-bound} and~\ref{assumption:interconnection-noise-matrices-QMI-local}, verifying the input-to-state stability of the noisy closed-loop system.}
For visualization purposes, we consider the closed-loop system without the process noise and the measurement noise.
As pointed out in Section~\ref{subsec:control-objective}, this implies the input-to-state stability with respect to the noise.
Furthermore, we generate $100$ random systems and apply Algorithms \ref{alg:decentralized-unknown} and \ref{alg:decentralized-diffusive-coupling}.
Both algorithms are feasible for all $100$ instances, and the average computation times per each DGU for these $100$ instances are $0.0023$s for Algorithm \ref{alg:decentralized-unknown} and $0.0018$s for Algorithm \ref{alg:decentralized-diffusive-coupling}.

\begin{table}
	\centering
	\begin{tabular}{lll}
		\hline \\[-0.2cm]
		Parameter & Symbol & Value \\ \\[-0.25cm]
		\hline \\ [-0.2cm]
		Internal Resistance & $R$ & $0.2 \pm 0.1 \; \Omega$ \\
		Internal Inductance & $L$ & $0.5 \pm 0.05 \; \revised{mH}$ \\
		Internal Capacitance & $C$ & $10 \pm 1.0 \; \revised{mF}$ \\
		Load Conductance & $Y$ & $0.2 \pm 0.02 \; S$ \\
		Line Resistance & $R_{ij}$ & $4 \pm 0.4 \; \Omega$ \\ \\ [-0.25cm]
		\hline
	\end{tabular}
	\caption{Electrical parameters used in the numerical experiments, comparable with parameters in \cite{riverso2014plug,tucci2016consensus}. The values for each DGU are drawn from uniform distributions indicated by the corresponding intervals.}
	\label{Tab:electricalparameters}
\end{table}

\begin{figure}
	\centering
	{{\includegraphics[width=0.45\columnwidth]{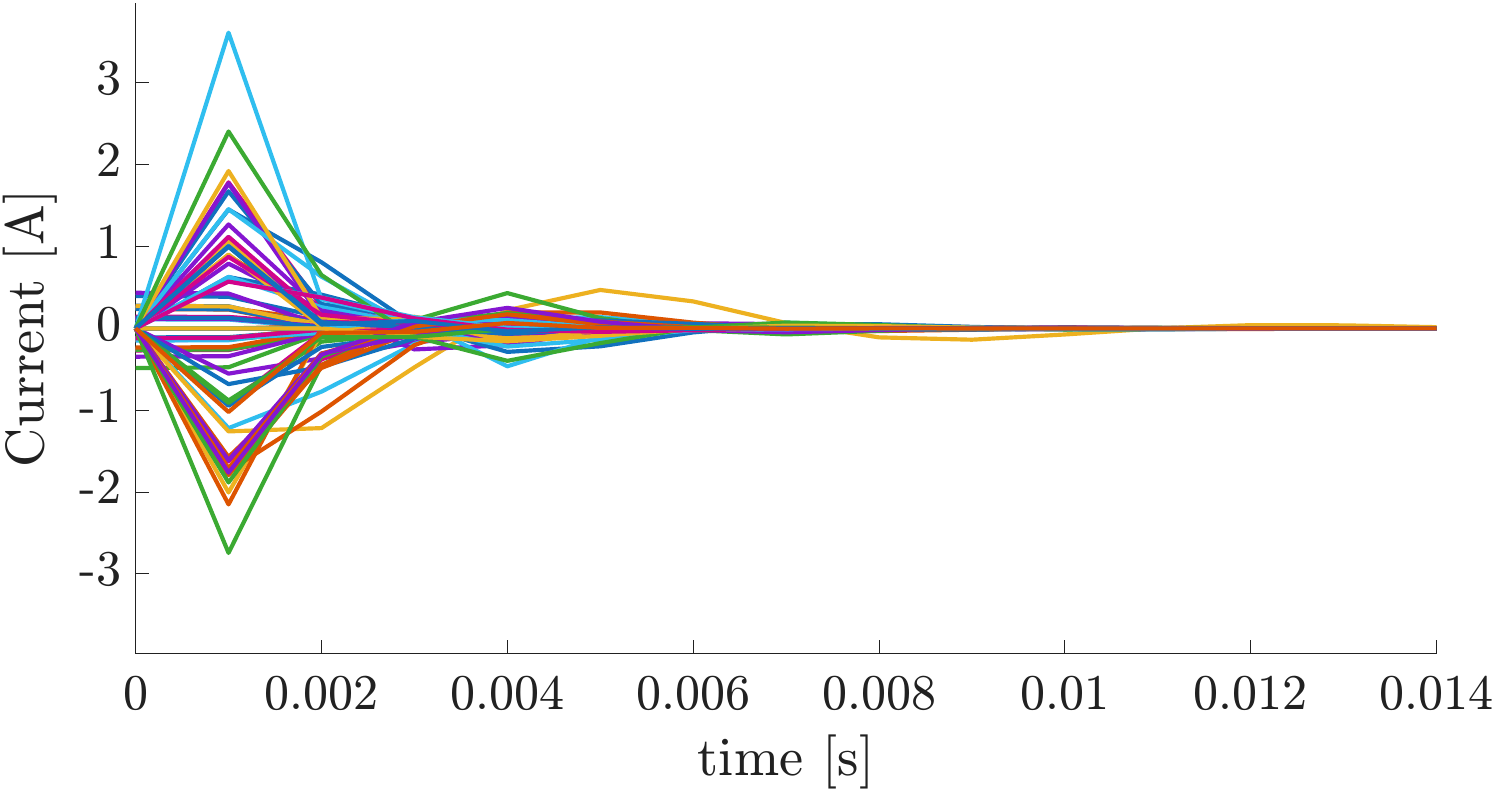} }}%
	{{\includegraphics[width=0.45\columnwidth]{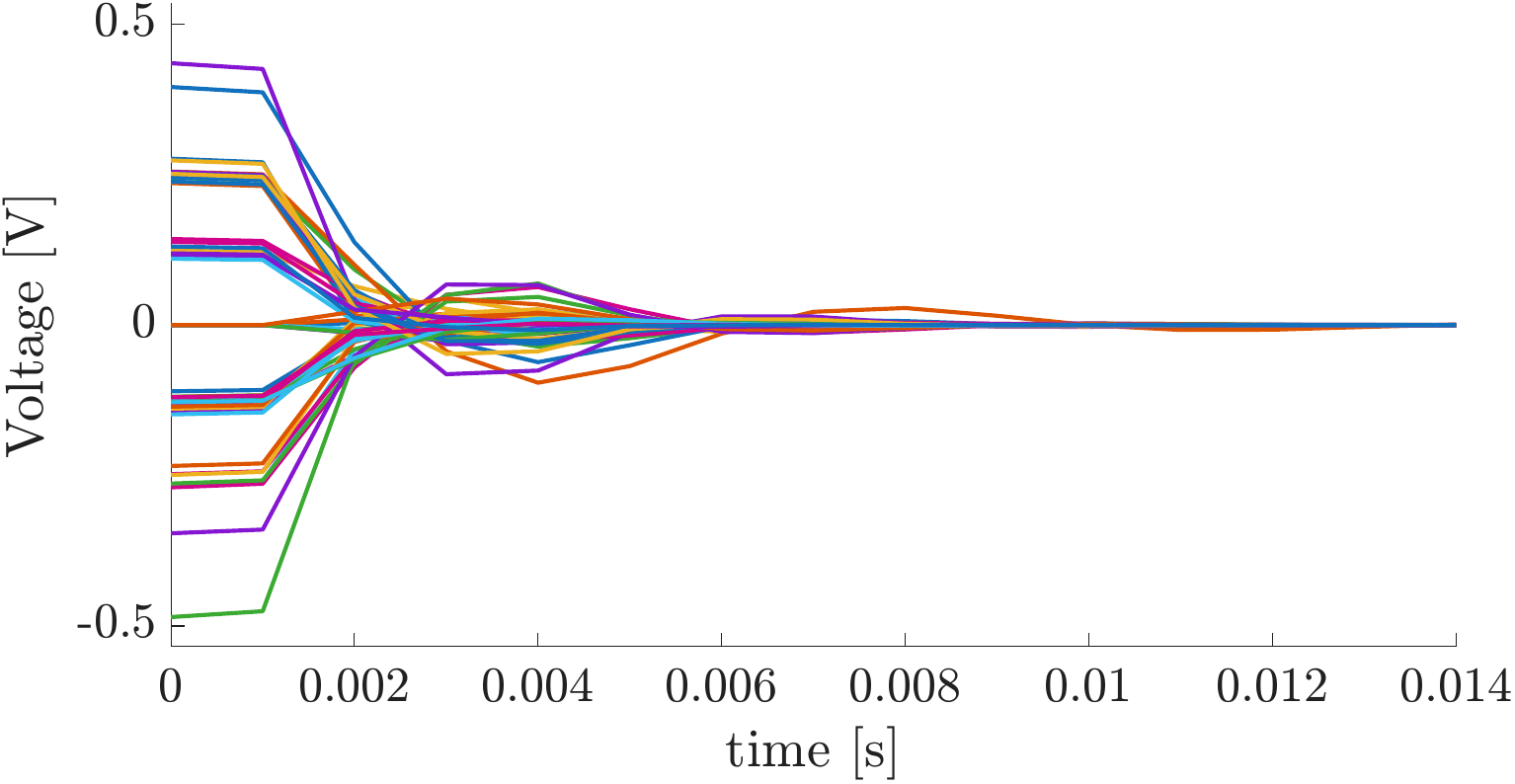} }}%
	\caption{\revised{The currents and the voltages of the DGUs after applying state-feedback controllers computed from Algorithm~\ref{alg:decentralized-unknown}.}}%
	\label{fig:numerical-experiments-alg4}%
\end{figure}

\begin{figure}
	\centering
	{{\includegraphics[width=0.45\columnwidth]{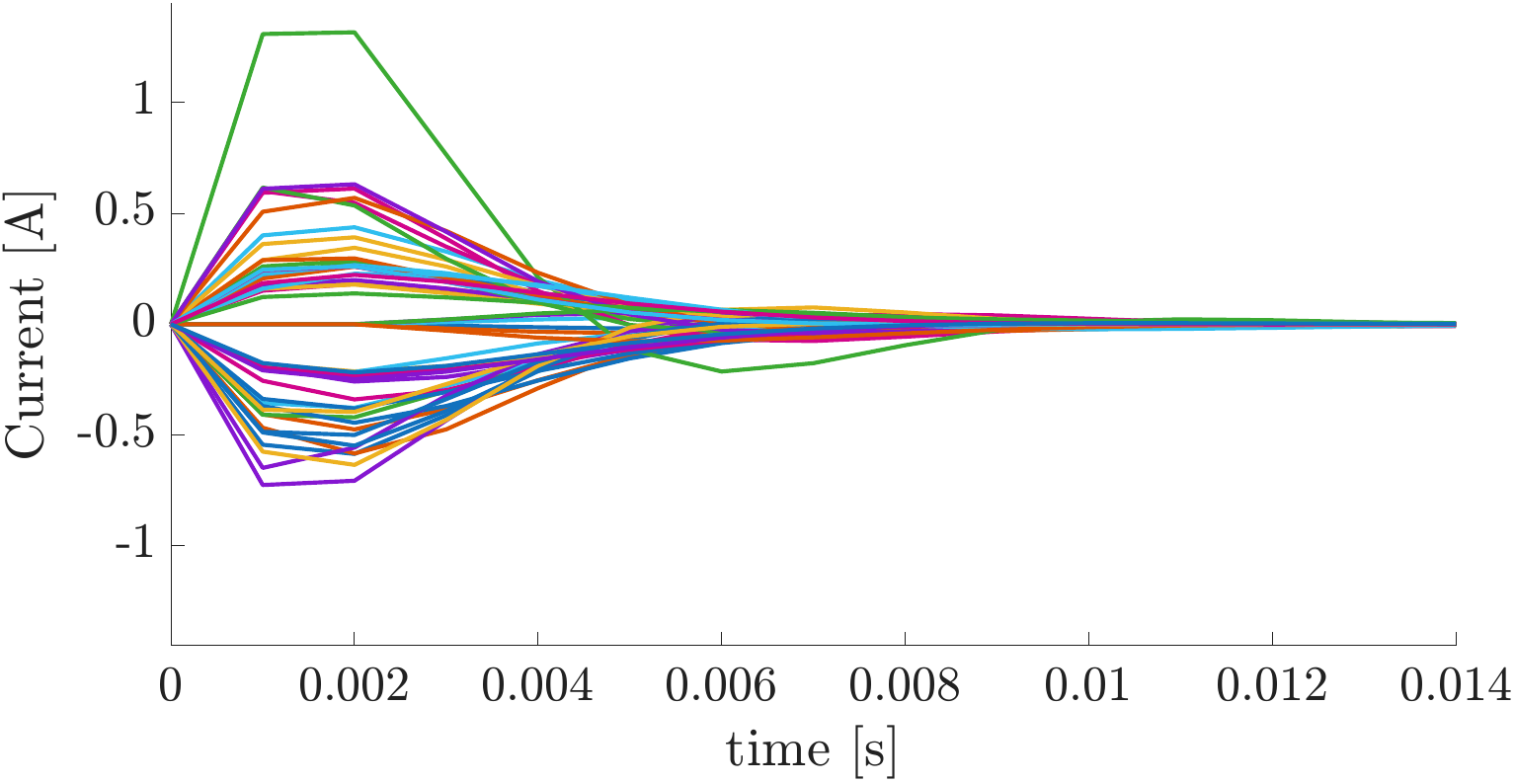} }}%
	{{\includegraphics[width=0.45\columnwidth]{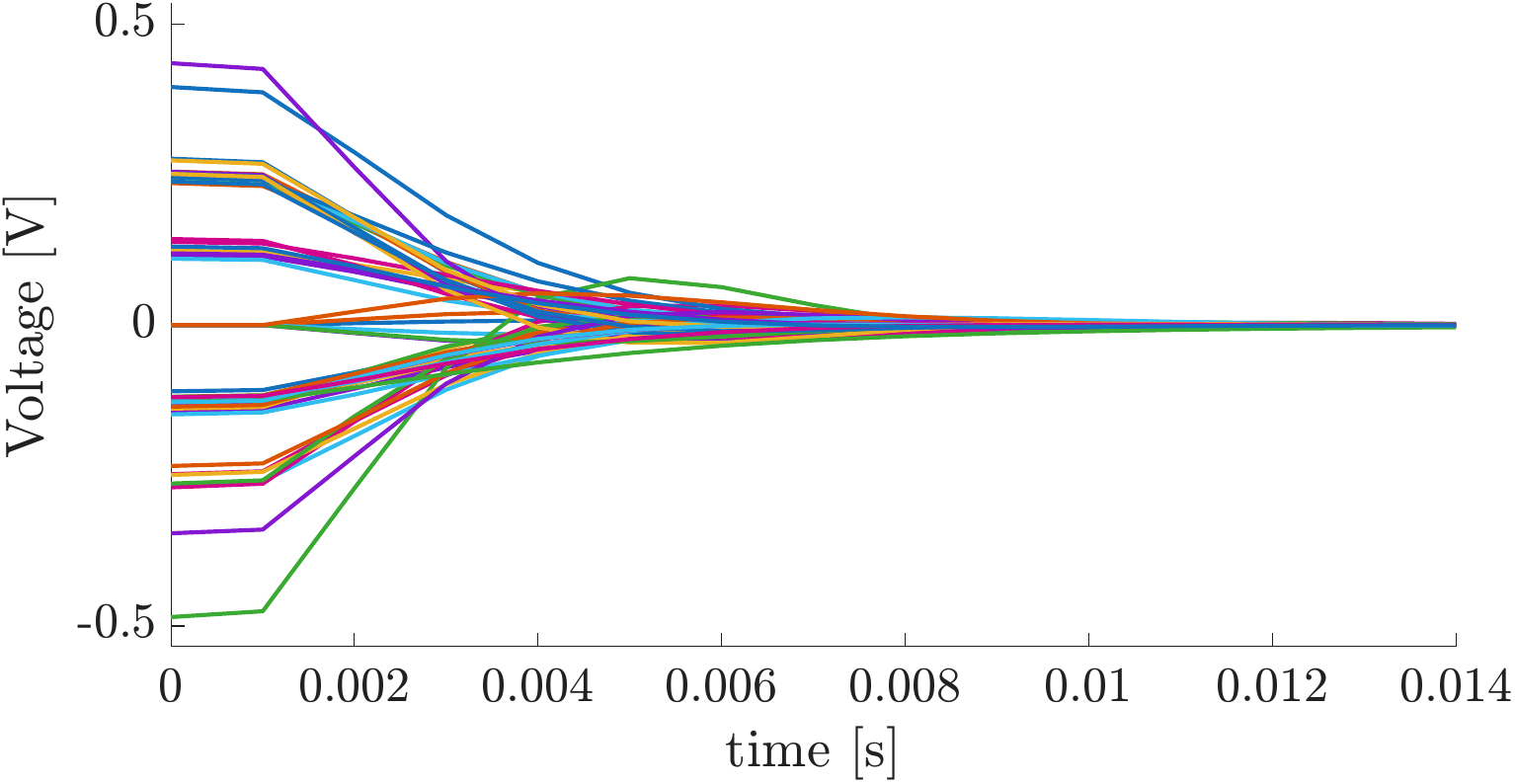} }}%
	\caption{\revised{The currents and the voltages of the DGUs after applying state-feedback controllers computed from Algorithm~\ref{alg:decentralized-diffusive-coupling}.}}%
	\label{fig:numerical-experiments-alg5}%
\end{figure}

\begin{figure}
	\centering
	{{\includegraphics[width=0.45\columnwidth]{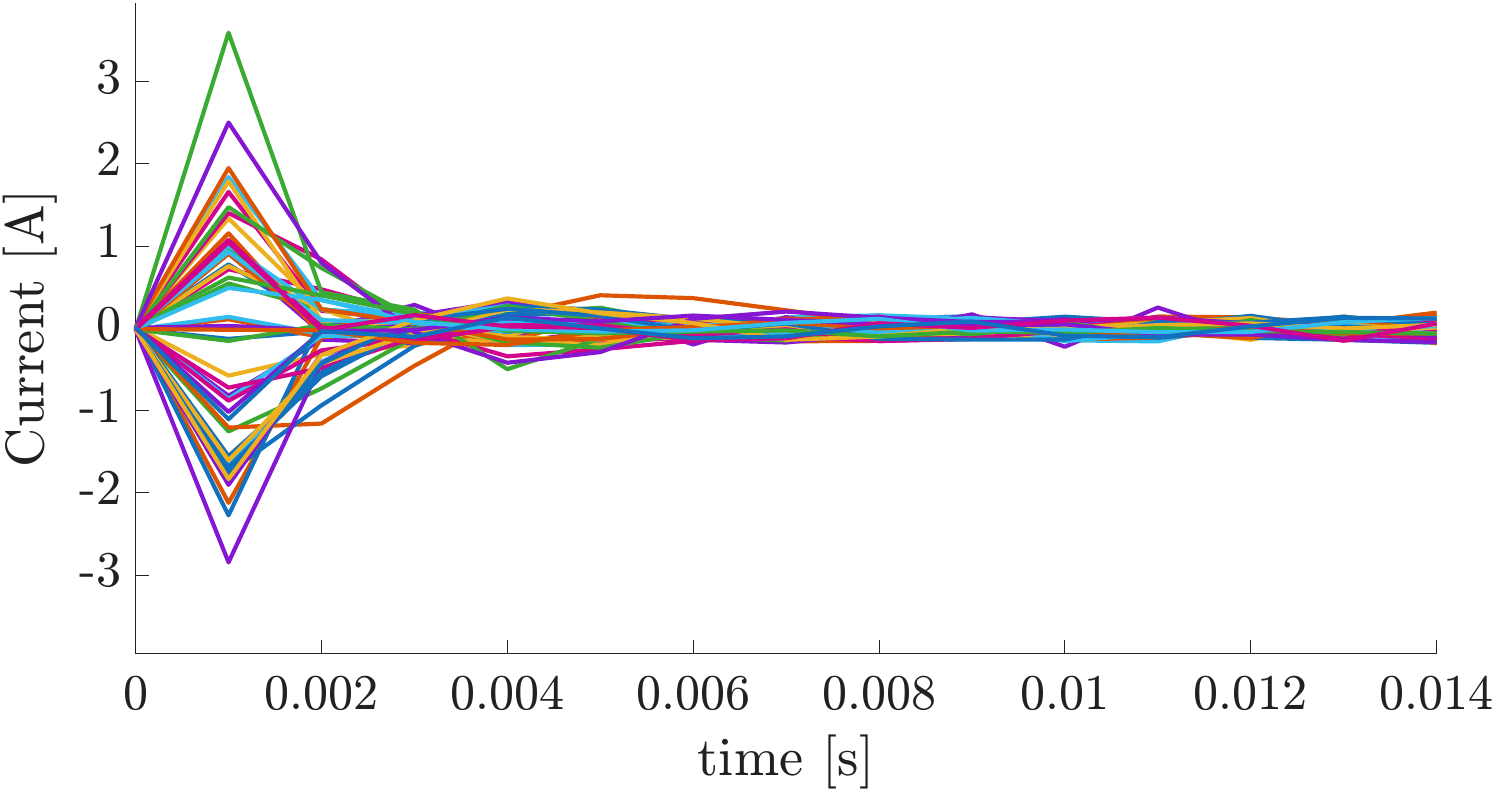} }}%
	{{\includegraphics[width=0.45\columnwidth]{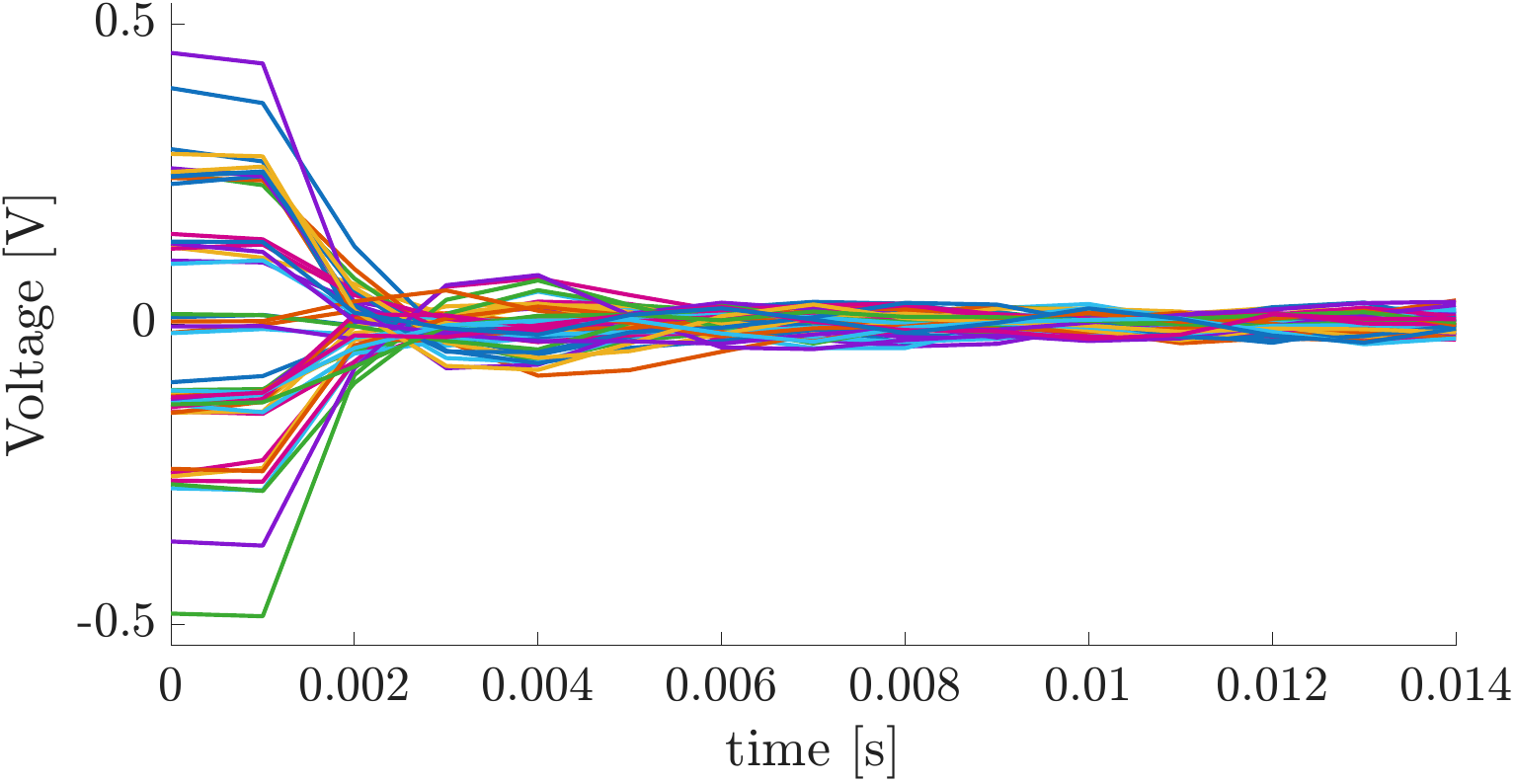} }}%
	\caption{\revised{The currents and the voltages of the DGUs after applying state-feedback controllers computed from Algorithm~\ref{alg:decentralized-unknown}, with the presence of noise.}}%
	\label{fig:iss}%
\end{figure}

\subsubsection*{Different experimental conditions}
Next, starting from the basic setup of $k = 50$ subsystems with $20$ non-ring interconnections,  data lengths $N_i = \tilde{N}_i = 50$ and noise parameters $\varepsilon_l = \varepsilon_g = 0.001$, we either (i) decrease the data lengths to $N_i = \tilde{N}_i = 20, 30, 40$, (ii) increase the noise parameters to \revised{$\varepsilon_l = \varepsilon_g = 0.0025, 0.005, 0.01$}, or (iii) increase the number of non-ring interconnections to $50, 75, 100$, and apply Algorithms \ref{alg:decentralized-unknown} and \ref{alg:decentralized-diffusive-coupling}.
For each condition, we generate $100$ random systems and compute the number of feasible instances, shown in Table \ref{Tab:feasibility-alg12}.

It can be seen that the number of feasible instances when applying Algorithm \ref{alg:decentralized-diffusive-coupling} decreases in all cases, whereas all instances remain feasible when applying Algorithm \ref{alg:decentralized-unknown}.
We conjecture that this is because in Algorithm \ref{alg:decentralized-diffusive-coupling}, the tightest upper bound of the weighted degree computed from \eqref{eq:degree-quadratic-bound} becomes larger and the feasible region of the stability condition \eqref{eq:diffusive-coupling-decentralized-condition-approximated} decreases, leading to infeasibility. 
Algorithm \ref{alg:decentralized-unknown}, on the other hand, does not require the computation of such an upper bound.

\revised{
	Increasing the noise parameters $\varepsilon_l$ and 
	$\varepsilon_g$ reduces the feasibility of Algorithm~\ref{alg:decentralized-diffusive-coupling}, whereas Algorithm~\ref{alg:decentralized-unknown} remains 
feasible in all instances.
	This suggests that larger noise levels complicate the synthesis problem, since the set of the system matrices and the interconnection matrices consistent with the noisy data gets enlarged.
}

\begin{table}
		\begin{subtable}[l]{.5\textwidth} \centering
		\begin{tabular}{lcccc}
			\hline \\[-0.2cm]
			Noise parameters & $0.001$  & $0.0025$ & $0.005$ & $0.01$ \\ \\[-0.25cm]
			\hline \\ [-0.2cm]
			Algorithm \ref{alg:decentralized-unknown} & $100\%$ & $100\%$ &$100\%$ & $100\%$  \\
			Algorithm \ref{alg:decentralized-diffusive-coupling} & $100\%$ & $85\%$ & $11\%$ & $0\%$ \\ \\ [-0.25cm]
			\hline
		\end{tabular}
		\caption{Feasible instances with respect to noise parameters $\varepsilon_l, \varepsilon_g$.}
		\label{Tab:2}
	\end{subtable}
	\begin{subtable}[l]{.5\textwidth} \centering
		\begin{tabular}{lcccc}
			\hline \\[-0.2cm]
			Data lengths & $20$ & $30$ & $40$ & $50$  \\ \\[-0.25cm]
			\hline \\ [-0.2cm]
			Algorithm \ref{alg:decentralized-unknown} & $100\%$ & $100\%$ & $100\%$ & $100\%$ \\
			Algorithm \ref{alg:decentralized-diffusive-coupling} & $0\%$ & $53\%$ & $97\%$ & $100\%$ \\ \\ [-0.25cm]
			\hline
		\end{tabular}
		\caption{Feasible instances with respect to data lengths $N_i, \tilde{N}_i$.}
		\label{Tab:3}
	\end{subtable}
	\begin{subtable}[l]{.5\textwidth}  \centering
		\begin{tabular}{lcccc}
			\hline \\[-0.2cm]
			Number of edges & $20$  & $50$ & $75$ & $100$ \\ \\[-0.25cm]
			\hline \\ [-0.2cm]
			Algorithm \ref{alg:decentralized-unknown} & $100\%$ & $100\%$ & $100\%$ & $100\%$ \\
			Algorithm \ref{alg:decentralized-diffusive-coupling} & $100\%$ & $94\%$ & $75\%$ &  $3\%$ \\ \\ [-0.25cm]
			\hline
		\end{tabular}
		\caption{Feasible instances with respect to non-ring interconnections.}
		\label{Tab:4}
	\end{subtable}
	\caption{Number of feasible instances out of $100$ random systems with respect to different parameters, applying Algorithms \ref{alg:decentralized-unknown} and \ref{alg:decentralized-diffusive-coupling}.}
	\label{Tab:feasibility-alg12}
\end{table}

\subsubsection*{Comparison to a centralized structured controller}
Structured control aims at designing controllers that satisfy particular structural constraints~\cite{jovanovic2016controller}.
Based on~\cite{miller2024data}, we implement a structured data-driven controller, which is a decentralized controller since it does not communicate with other controllers, but is designed in a centralized fashion.
We show that it requires more data and computation time than Algorithms \ref{alg:decentralized-unknown} and \ref{alg:decentralized-diffusive-coupling}.

From \eqref{eq:dgu-dynamics-discrete-time} and \eqref{interconnectionlines}, the global dynamics of the microgrid system can be written as 
\begin{equation*}
  \begin{aligned}
    x(t+1) & = A x(t) + B_{1} u(t) + B_{2} v(t), \\
		y(t) & = C x(t), \\
    v(t) &= M y(t),
  \end{aligned}
\end{equation*}
with global vectors $x(t),u(t),v(t),y(t)$ and global matrices $A,B_1,B_2,C$. 
By eliminating $v(t)$ and $y(t)$, we obtain
\begin{equation} \label{eq:global-system-microgrid}
  x(t+1) = (A + B_2 M C) x(t) + B_{1} u(t).
\end{equation}
Since we design a local controller of the form $u_i = K_i x_i$ for all $i\in \{1,\dots,k\}$, we design in a centralized fashion a structured state-feedback controller $u = K x$ such that $K = \operatorname{diag}(K_1,\dots,K_k) \in \mathbb{R}^{k \times 2k}$.
Since the system matrices are unknown, we assume that we can collect $N_g$-length data trajectories $\{x(t)\}_{t=0}^{N_g}$ and $\{u(t)\}_{t=0}^{N_g-1}$ from \eqref{eq:global-system-microgrid}.
Then, we adopt from \cite{miller2024data} a data-driven LMI condition, where the decision variables are restricted to be block diagonal.
We consider $k = 50$ subsystems with $20$ non-ring interconnections and check the feasibility of the LMI condition, with respect to different data lengths $N_g = 100,400,600,800$.
A number of $100$ random system instances are generated, and we compute the number of feasible instances and the average computation time of the data-driven structured controller, shown in Table~\ref{Tab:5}.
The centralized structured control approach has advantages such that it does not require the information of the interconnection structure and can be generalized to other sparsity patterns of the state-feedback gain.
However, Table \ref{Tab:5} shows that significantly more data is required to achieve feasibility, compared to Algorithms \ref{alg:decentralized-unknown} and \ref{alg:decentralized-diffusive-coupling} that require only data lengths of $50$.
Moreover, the computation time required is considerably more than that required for Algorithms \ref{alg:decentralized-unknown} and \ref{alg:decentralized-diffusive-coupling}, which are $0.0023$s and $0.0018$s, respectively.

\begin{table}
		\centering
		\begin{tabular}{lcccc}
			\hline \\[-0.2cm]
			Data length & $100$ & $400$ & $600$ & $800$ \\ \\[-0.25cm]
			\hline \\ [-0.2cm]
			Feasible instances & $0\%$ & $0\%$ & $87\%$ & $100\%$ \\
			Computation time (s)& $23.58$ & $25.66$ & $35.08$ & $ 23.94$ \\ \\ [-0.25cm]
			\hline
		\end{tabular}
		\caption{Feasible instances out of 100 random systems and the average computation time with respect to different data lengths, when applying structured data-driven control~\cite{miller2024data}.}
		\label{Tab:5}
	\end{table}

\section{Conclusions}
\label{sec:conclusions_and_discussions}
We present a data-driven decentralized control algorithm for unknown interconnected systems.
We derive the local controller design LMI \eqref{eq:local-controller-design-lmi} and the decentralized LMI \eqref{eq:local-lmi-condition-for-data-driven-stability}, both based on dissipativity theory.
Since the supply rate matrices $(F_i,G_i,H_i)$ appear linearly in both LMIs, we are able to treat them as decision variables to solve them at once, yielding local controllers that guarantee asymptotic stability of the global system.
Moreover, for the diffusive coupling case, we develop a quadratic programming formulation \eqref{eq:degree-quadratic-bound} to obtain the upper bound of the weighted degree, which is combined with the stability condition \eqref{eq:diffusive-coupling-decentralized-condition}.
Again, dissipativity theory serves as a pipeline to combine multiple conditions to yield a data-driven decentralized control algorithm.
Numerical examples are shown in the context of microgrids, showcasing the effectiveness and the scalability of the proposed control algorithms, compared to a centralized structured data-driven controller.
\revised{
Future research directions include the online adaptation of the proposed algorithms, incorporating performance metrics to the LMIs to improve the transient performance and the robustness against noise, and the consideration of broader classes of systems including nonlinear local dynamics and dynamic interconnection.
}

\section*{References}
\bibliographystyle{ieeetr}
\bibliography{reffile}

\begin{thebibliography}{10}

\bibitem{bakule2008decentralized}
L.~Bakule, ``Decentralized control: An overview,'' {\em {Annual Reviews in Control}}, vol.~32, no.~1, pp.~87--98, 2008.

\bibitem{siljak2013decentralized}
D.~D. Siljak, {\em Decentralized control of complex systems}.
\newblock Courier Corporation, 2013.

\bibitem{bakule2014decentralized}
L.~Bakule, ``Decentralized control: Status and outlook,'' {\em {Annual Reviews in Control}}, vol.~38, no.~1, pp.~71--80, 2014.

\bibitem{van2000l2}
A.~van~der Schaft, {\em L2-gain and passivity techniques in nonlinear control}.
\newblock Springer, 2000.

\bibitem{willems1972dissipative1}
J.~C. Willems, ``Dissipative dynamical systems part {I}: General theory,'' {\em Archive for rational mechanics and analysis}, vol.~45, no.~5, pp.~321--351, 1972.

\bibitem{chopra2006passivity}
N.~Chopra and M.~W. Spong, ``Passivity-based control of multi-agent systems,'' {\em Advances in robot control: from everyday physics to human-like movements}, pp.~107--134, 2006.

\bibitem{arcak2016networks}
M.~Arcak, C.~Meissen, and A.~Packard, {\em Networks of dissipative systems: compositional certification of stability, performance, and safety}.
\newblock Springer, 2016.

\bibitem{ortega1998passivity}
R.~Ortega, A.~Loria, P.~J. Nicklasson, and H.~Sira-Ramirez, {\em Passivity-based Control of Euler-Lagrange Systems}.
\newblock Springer, 1998.

\bibitem{byrnes1993discrete}
C.~I. Byrnes and W.~Lin, ``Discrete-time lossless systems, feedback equivalence and passivity,'' in {\em Proceedings of 32nd IEEE Conference on Decision and Control}, pp.~1775--1781, IEEE, 1993.

\bibitem{byrnes1994losslessness}
C.~I. Byrnes and W.~Lin, ``Losslessness, feedback equivalence, and the global stabilization of discrete-time nonlinear systems,'' {\em IEEE Transactions on automatic control}, vol.~39, no.~1, pp.~83--98, 1994.

\bibitem{navarro2002dissipativity}
E.~M. Navarro-L{\'o}pez, H.~Sira-Ram{\'\i}rez, and E.~Fossas-Colet, ``Dissipativity and feedback dissipativity properties of general nonlinear discrete-time systems,'' {\em European journal of control}, vol.~8, no.~3, pp.~265--274, 2002.

\bibitem{moreschini2024generalized}
A.~Moreschini, M.~Bin, A.~Astolfi, and T.~Parisini, ``A generalized passivity theory over abstract time domains,'' {\em IEEE Transactions on Automatic Control}, vol.~70, no.~1, pp.~2--17, 2024.

\bibitem{mccourt2012stability}
M.~J. McCourt and P.~J. Antsaklis, ``Stability of interconnected switched systems using {QSR} dissipativity with multiple supply rates,'' in {\em 2012 American Control Conference (ACC)}, pp.~4564--4569, IEEE, 2012.

\bibitem{aboudonia2021passivity}
A.~Aboudonia, A.~Martinelli, and J.~Lygeros, ``Passivity-based decentralized control for discrete-time large-scale systems,'' in {\em 2021 American Control Conference (ACC)}, pp.~2037--2042, IEEE, 2021.

\bibitem{martinelli2023interconnection}
A.~Martinelli, A.~Aboudonia, and J.~Lygeros, ``{Interconnection of $(Q, S, R)$-Dissipative Systems in Discrete Time},'' {\em arXiv preprint arXiv:2311.08088}, 2023.

\bibitem{ljung1998system}
L.~Ljung, ``System identification,'' in {\em Signal analysis and prediction}, pp.~163--173, Springer, 1998.

\bibitem{gevers2005identification}
M.~Gevers, ``{Identification for control: From the early achievements to the revival of experiment design},'' {\em European journal of control}, vol.~11, no.~4-5, pp.~335--352, 2005.

\bibitem{willems2005note}
J.~C. Willems, P.~Rapisarda, I.~Markovsky, and B.~L. De~Moor, ``A note on persistency of excitation,'' {\em Systems \& Control Letters}, vol.~54, no.~4, pp.~325--329, 2005.

\bibitem{willems1997introduction}
J.~C. Willems and J.~W. Polderman, {\em Introduction to mathematical systems theory: a behavioral approach}, vol.~26.
\newblock Springer Science \& Business Media, 1997.

\bibitem{markovsky2008data}
I.~Markovsky and P.~Rapisarda, ``Data-driven simulation and control,'' {\em International Journal of Control}, vol.~81, no.~12, pp.~1946--1959, 2008.

\bibitem{yang2015data}
H.~Yang and S.~Li, ``A data-driven predictive controller design based on reduced {Hankel} matrix,'' in {\em 2015 10th Asian Control Conference (ASCC)}, pp.~1--7, IEEE, 2015.

\bibitem{coulson2019data}
J.~Coulson, J.~Lygeros, and F.~D{\"o}rfler, ``Data-enabled predictive control: In the shallows of the {DeePC},'' in {\em 2019 18th European Control Conference (ECC)}, pp.~307--312, IEEE, 2019.

\bibitem{de2019formulas}
C.~De~Persis and P.~Tesi, ``Formulas for data-driven control: Stabilization, optimality, and robustness,'' {\em IEEE Transactions on Automatic Control}, vol.~65, no.~3, pp.~909--924, 2020.

\bibitem{berberich2020data}
J.~Berberich, J.~K{\"o}hler, M.~A. M{\"u}ller, and F.~Allg{\"o}wer, ``Data-driven model predictive control with stability and robustness guarantees,'' {\em IEEE Transactions on Automatic Control}, vol.~66, no.~4, pp.~1702--1717, 2020.

\bibitem{van2020data}
H.~J. van Waarde, M.~K. Camlibel, and M.~Mesbahi, ``From noisy data to feedback controllers: {Nonconservative} design via a matrix {S-Lemma},'' {\em IEEE Transactions on Automatic Control}, vol.~67, no.~1, pp.~162--175, 2022.

\bibitem{maupong2017lyapunov}
T.~Maupong, J.~C. Mayo-Maldonado, and P.~Rapisarda, ``On {Lyapunov} functions and data-driven dissipativity,'' {\em IFAC-PapersOnLine}, vol.~50, no.~1, pp.~7783--7788, 2017.

\bibitem{koch2021provably}
A.~Koch, J.~Berberich, and F.~Allg{\"o}wer, ``Provably robust verification of dissipativity properties from data,'' {\em IEEE Transactions on Automatic Control}, vol.~67, no.~8, pp.~4248--4255, 2021.

\bibitem{rosa2021one}
T.~E. Rosa and B.~Jayawardhana, ``On the one-shot data-driven verification of dissipativity of {LTI} systems with general quadratic supply rate function,'' in {\em 2021 European Control Conference (ECC)}, pp.~1291--1296, IEEE, 2021.

\bibitem{van2022data}
H.~J. van Waarde, M.~K. Camlibel, P.~Rapisarda, and H.~L. Trentelman, ``Data-driven dissipativity analysis: application of the matrix {S-lemma},'' {\em IEEE Control Systems Magazine}, vol.~42, no.~3, pp.~140--149, 2022.

\bibitem{nguyen2024synthesis}
E.~T. Nguyen and H.~J. van Waarde, ``{Synthesis of Dissipative Systems Using Input-State Data},'' in {\em 2024 European Control Conference (ECC)}, pp.~2959--2964, IEEE, 2024.

\bibitem{kristovic2024data}
P.~Kristovi{\'c} and A.~Joki{\'c}, ``{Data-Driven State-Feedback Controller Synthesis for Dissipativity: A Dualization-Based Approach},'' in {\em 2024 American Control Conference (ACC)}, pp.~1219--1224, IEEE, 2024.

\bibitem{tanaka2024algebraic}
Y.~Tanaka, O.~Kaneko, and T.~Sueyoshi, ``Algebraic approach to synthesis of data-driven control design for dissipativity,'' {\em SICE Journal of Control, Measurement, and System Integration}, vol.~17, no.~1, pp.~247--255, 2024.

\bibitem{baggio2021data}
G.~Baggio, D.~S. Bassett, and F.~Pasqualetti, ``Data-driven control of complex networks,'' {\em Nature communications}, vol.~12, no.~1, p.~1429, 2021.

\bibitem{wang2023data}
X.~Wang, J.~Sun, G.~Wang, F.~Allg{\"o}wer, and J.~Chen, ``Data-driven control of distributed event-triggered network systems,'' {\em IEEE/CAA Journal of Automatica Sinica}, vol.~10, no.~2, pp.~351--364, 2023.

\bibitem{celi2023distributed}
F.~Celi, G.~Baggio, and F.~Pasqualetti, ``Distributed data-driven control of network systems,'' {\em IEEE Open Journal of Control Systems}, vol.~2, pp.~93--107, 2023.

\bibitem{akbarzadeh2024data}
O.~Akbarzadeh, A.~Nejati, and A.~Lavaei, ``{From Data to Control: A Formal Compositional Framework for Large-Scale Interconnected Networks},'' {\em arXiv preprint arXiv:2409.12469}, 2024.

\bibitem{liao2024data}
J.~Liao, S.~Lu, T.~Wang, and W.~Xiang, ``{Data-Driven Decentralized Control Design for Discrete-Time Large-Scale Systems},'' {\em arXiv preprint arXiv:2411.10243}, 2024.

\bibitem{van2023quadratic}
H.~J. van Waarde, M.~K. Camlibel, J.~Eising, and H.~L. Trentelman, ``Quadratic matrix inequalities with applications to data-based control,'' {\em SIAM Journal on Control and Optimization}, vol.~61, no.~4, pp.~2251--2281, 2023.

\bibitem{burohman2023data}
A.~M. Burohman, ``From data to reduced-order models of complex dynamical systems,'' {\em Ph.D. dissertation, University of Groningen}, 2023.

\bibitem{bullo2018lectures}
F.~Bullo, J.~Cort{\'e}s, F.~D{\"o}rfler, and S.~Mart{\'\i}nez, {\em Lectures on network systems}, vol.~1.
\newblock CreateSpace, 2018.

\bibitem{berberich2020robust}
J.~Berberich, A.~Koch, C.~W. Scherer, and F.~Allg{\"o}wer, ``Robust data-driven state-feedback design,'' in {\em 2020 American Control Conference (ACC)}, pp.~1532--1538, IEEE, 2020.

\bibitem{jiang2001input}
Z.-P. Jiang and Y.~Wang, ``Input-to-state stability for discrete-time nonlinear systems,'' {\em Automatica}, vol.~37, no.~6, pp.~857--869, 2001.

\bibitem{horn2012matrix}
R.~A. Horn and C.~R. Johnson, {\em Matrix analysis}.
\newblock Cambridge university press, 2012.

\bibitem{boyd2004convex}
S.~Boyd and L.~Vandenberghe, {\em Convex optimization}.
\newblock Cambridge university press, 2004.

\bibitem{brogliato2007dissipative}
B.~Brogliato, R.~Lozano, B.~Maschke, O.~Egeland, {\em et~al.}, ``Dissipative systems analysis and control,'' {\em Theory and applications}, vol.~2, pp.~2--5, 2007.

\bibitem{zhu2014passivity}
F.~Zhu, M.~Xia, and P.~J. Antsaklis, ``Passivity analysis and passivation of feedback systems using passivity indices,'' in {\em 2014 American Control Conference}, pp.~1833--1838, IEEE, 2014.

\bibitem{lofberg2004yalmip}
J.~Lofberg, ``Yalmip: A toolbox for modeling and optimization in matlab,'' in {\em 2004 IEEE international conference on robotics and automation (IEEE Cat. No. 04CH37508)}, pp.~284--289, IEEE, 2004.

\bibitem{mosek}
{MOSEK ApS}, {\em The MOSEK optimization toolbox for MATLAB manual. Version 9.0.}, 2019.

\bibitem{riverso2014plug}
S.~Riverso, F.~Sarzo, and G.~Ferrari-Trecate, ``Plug-and-play voltage and frequency control of islanded microgrids with meshed topology,'' {\em IEEE Transactions on Smart Grid}, vol.~6, no.~3, pp.~1176--1184, 2014.

\bibitem{tucci2016consensus}
M.~Tucci, L.~Meng, J.~M. Guerrero, and G.~Ferrari-Trecate, ``A consensus-based secondary control layer for stable current sharing and voltage balancing in {DC} microgrids,'' {\em arXiv preprint arXiv:1603.03624}, 2016.

\bibitem{jovanovic2016controller}
M.~R. Jovanovi{\'c} and N.~K. Dhingra, ``Controller architectures: Tradeoffs between performance and structure,'' {\em European Journal of Control}, vol.~30, pp.~76--91, 2016.

\bibitem{miller2024data}
J.~Miller, J.~Eising, F.~D{\"o}rfler, and R.~S. Smith, ``Data-driven structured robust control of linear systems,'' {\em arXiv preprint arXiv:2411.11542}, 2024.

\end{thebibliography}

\begin{IEEEbiography}[{\includegraphics[width=1in,height=1.25in,clip,keepaspectratio]{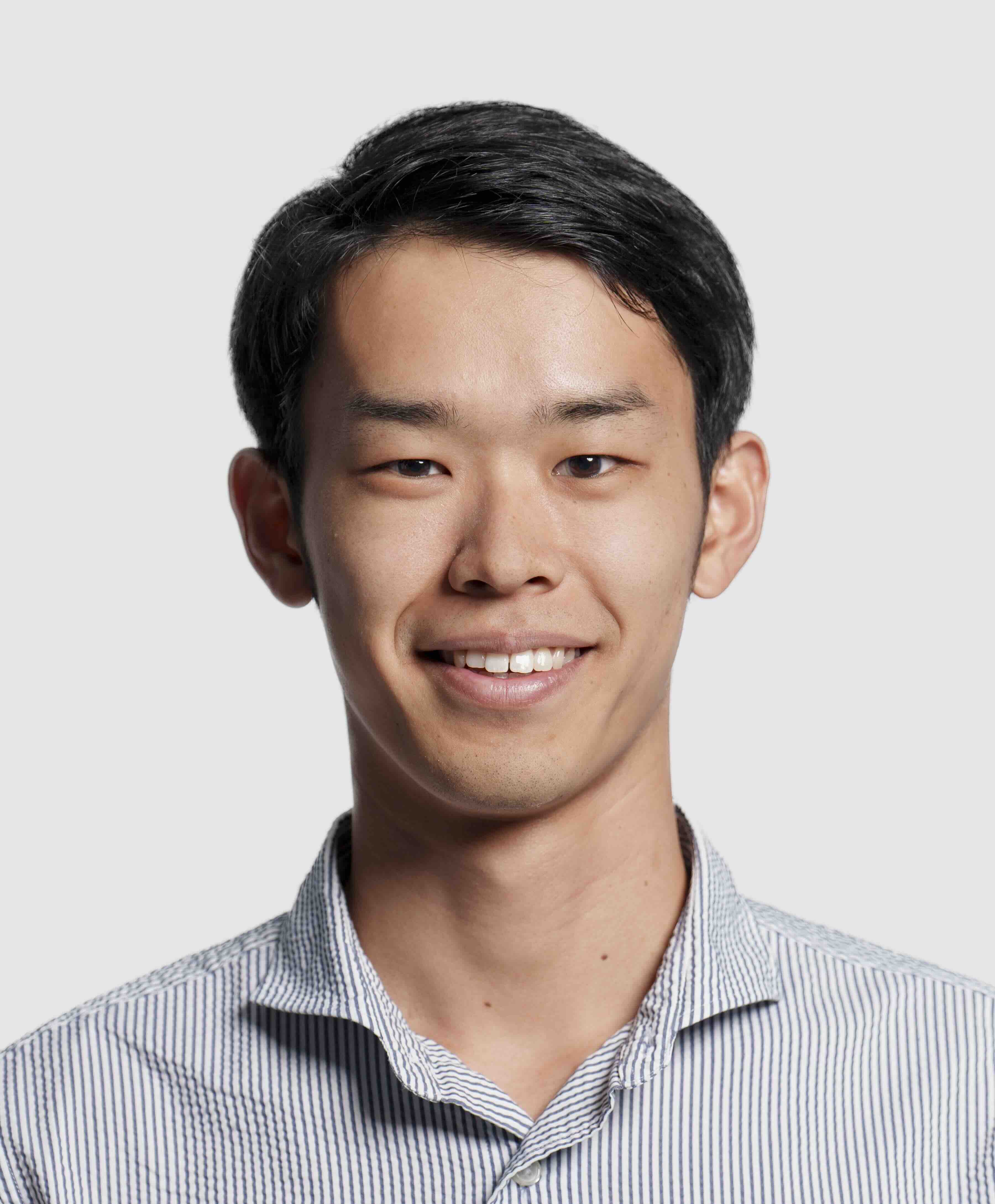}}]{Taiki Nakano}
	is a Ph.D. student at the Automatic Control Laboratory, ETH Zürich, Switzerland and the Learning and Dynamical Systems Group at the Max Planck Institute for Intelligent Systems, Tübingen, Germany, affiliated with the Max Planck ETH Center for Learning Systems (CLS).
	He received his B.Eng. degree in 2022 and his M.Sc. degree in 2024, both from the University of Tokyo, Japan.
	His research interests include system theory, optimization and machine learning.
\end{IEEEbiography}

\begin{IEEEbiography}[{\includegraphics[width=1in,height=1.25in,clip,keepaspectratio]{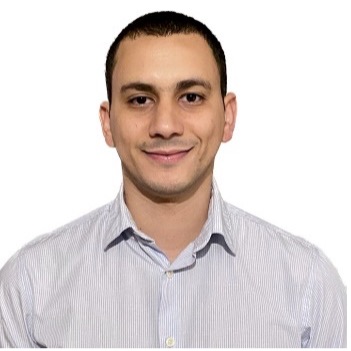}}]{Ahmed Aboudonia} 
	received a B.Sc. in Mechatronics Engineering from the German University in Cairo and an M.Sc. in Control Engineering from Sapienza University of Rome. He completed his PhD with the Automatic Control Lab at ETH Zurich. Currently, he is a postdoctoral researcher with the Advanced Control Research Lab at the University of Illinois Urbana-Champaign. Previously, he has conducted research at the German Aerospace Center and the Technical University of Darmstadt. His research interests include adaptive, data-driven, and learning-based control, with applications to robotics and energy systems.
\end{IEEEbiography}

\begin{IEEEbiography}[{\includegraphics[width=1in,height=1.25in,clip,keepaspectratio]{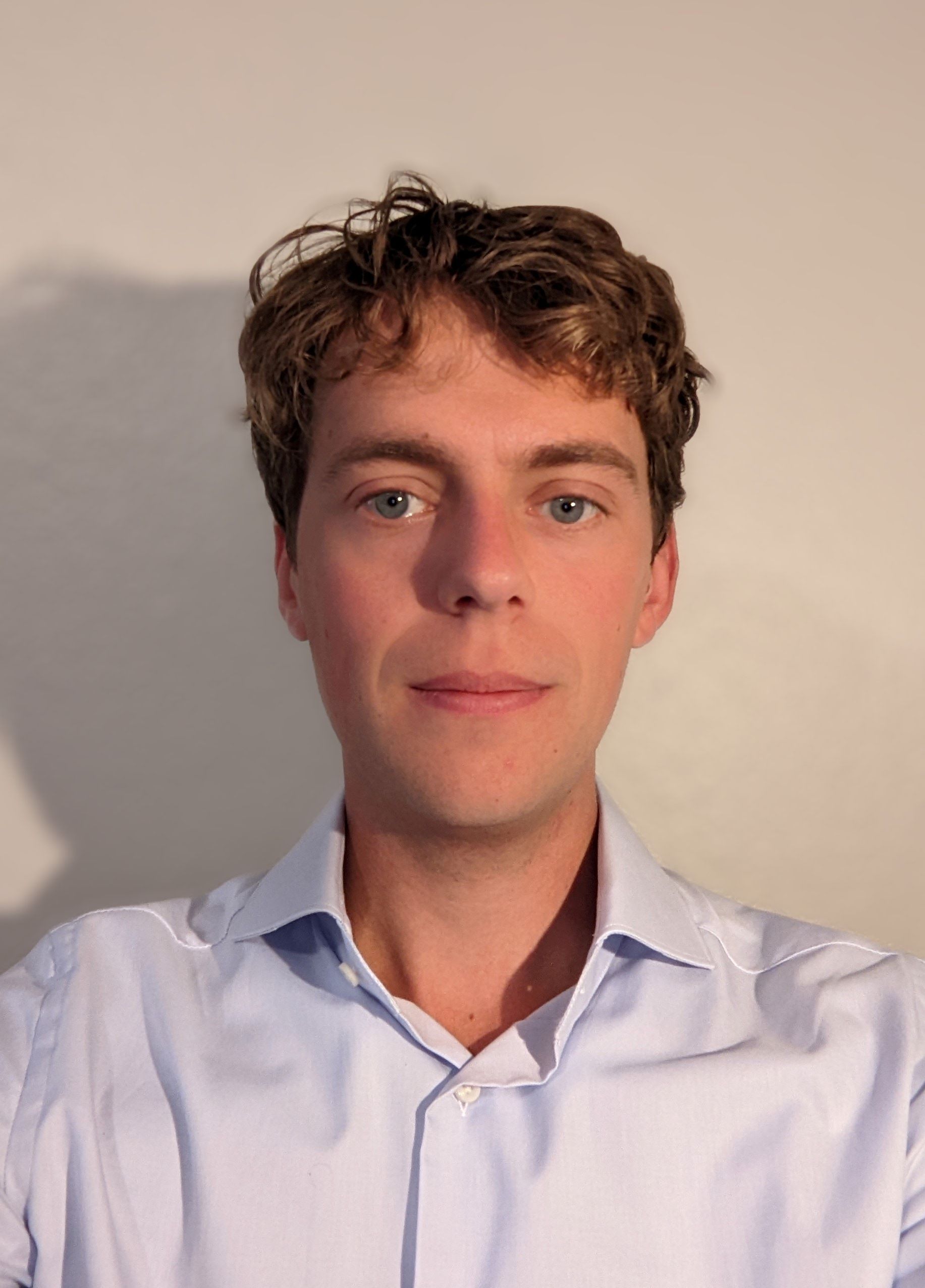}}]{Jaap Eising} 
  is a senior scientist at the Automatic Control Laboratory, ETH Zürich. Before that, he was at the Department of Mechanical and Aerospace Engineering at the University of California, San Diego. He attained the Ph.D. degree at the University of Groningen in 2021, after obtaining the M.Sc. in Mathematics at the same university in 2017. His research interests include data-driven control and optimization, behavioral methods, and algorithm design.
\end{IEEEbiography}

\begin{IEEEbiography}[{\includegraphics[width=1in,height=1.25in,clip,keepaspectratio]{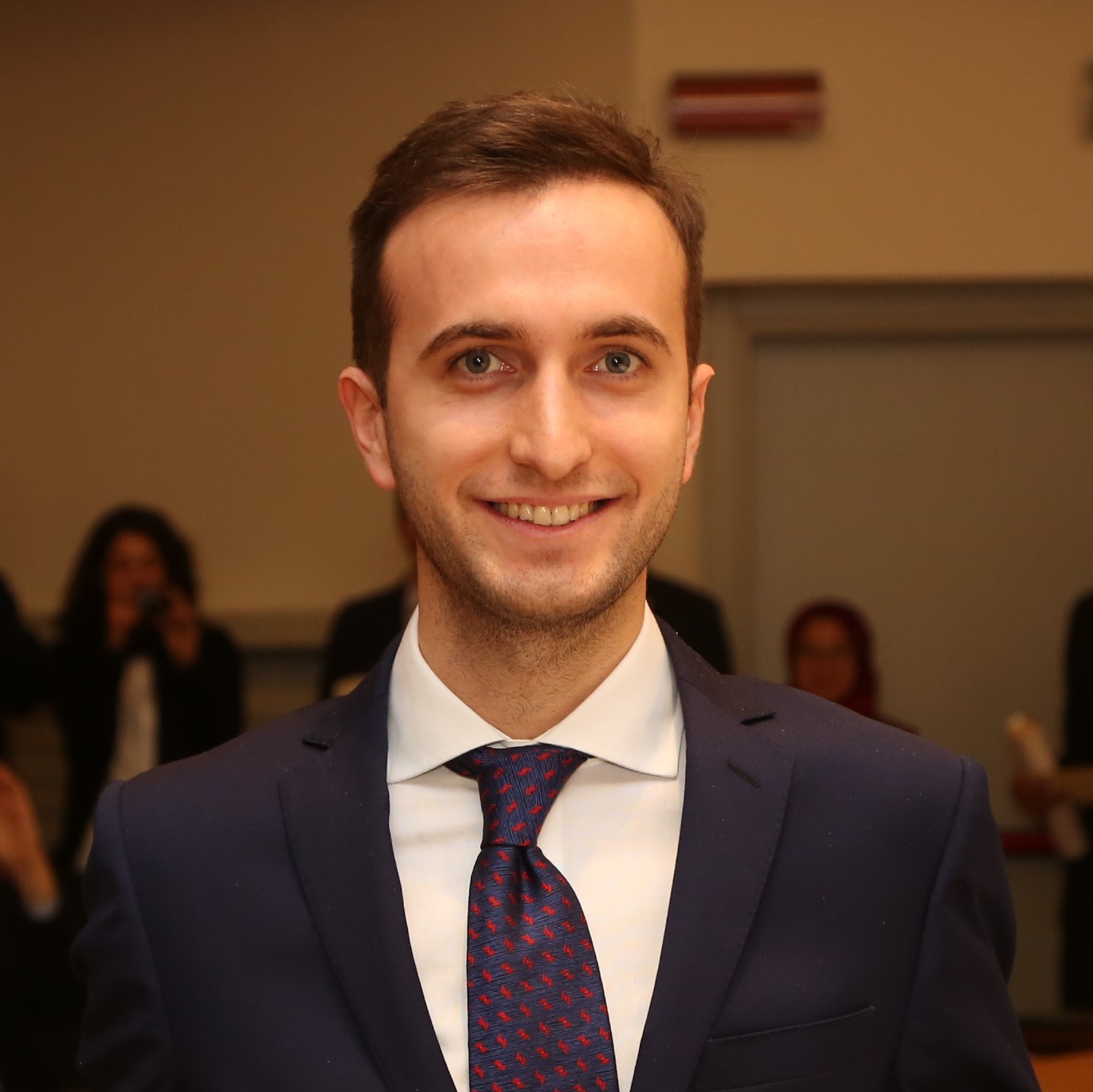}}]{Andrea Martinelli}
	is a Postdoctoral Researcher at the Automatic Control Laboratory, ETH Zürich, Switzerland, and Program Manager of the Certificate of Advanced Studies (CAS) ETH in Automation. He received the B.Sc. degree in management engineering and the M.Sc. degree in control engineering from Politecnico di Milano, Italy, in 2015 and 2017, respectively. In 2024, he obtained a PhD degree at the Automatic Control Laboratory, ETH Zürich, for which he received the ETH Medal for outstanding doctoral thesis. In 2017, he conducted his M.Sc. thesis at the Laboratoire d'Automatique, EPF Lausanne, Switzerland. In 2018, he was a Research Assistant with the Systems \& Control Group at Politecnico di Milano. His research interests include data-driven control, dynamic programming, and dissipativity theory for interconnected systems.
\end{IEEEbiography}

\begin{IEEEbiography}[{\includegraphics[width=1in,height=1.25in,clip,keepaspectratio]{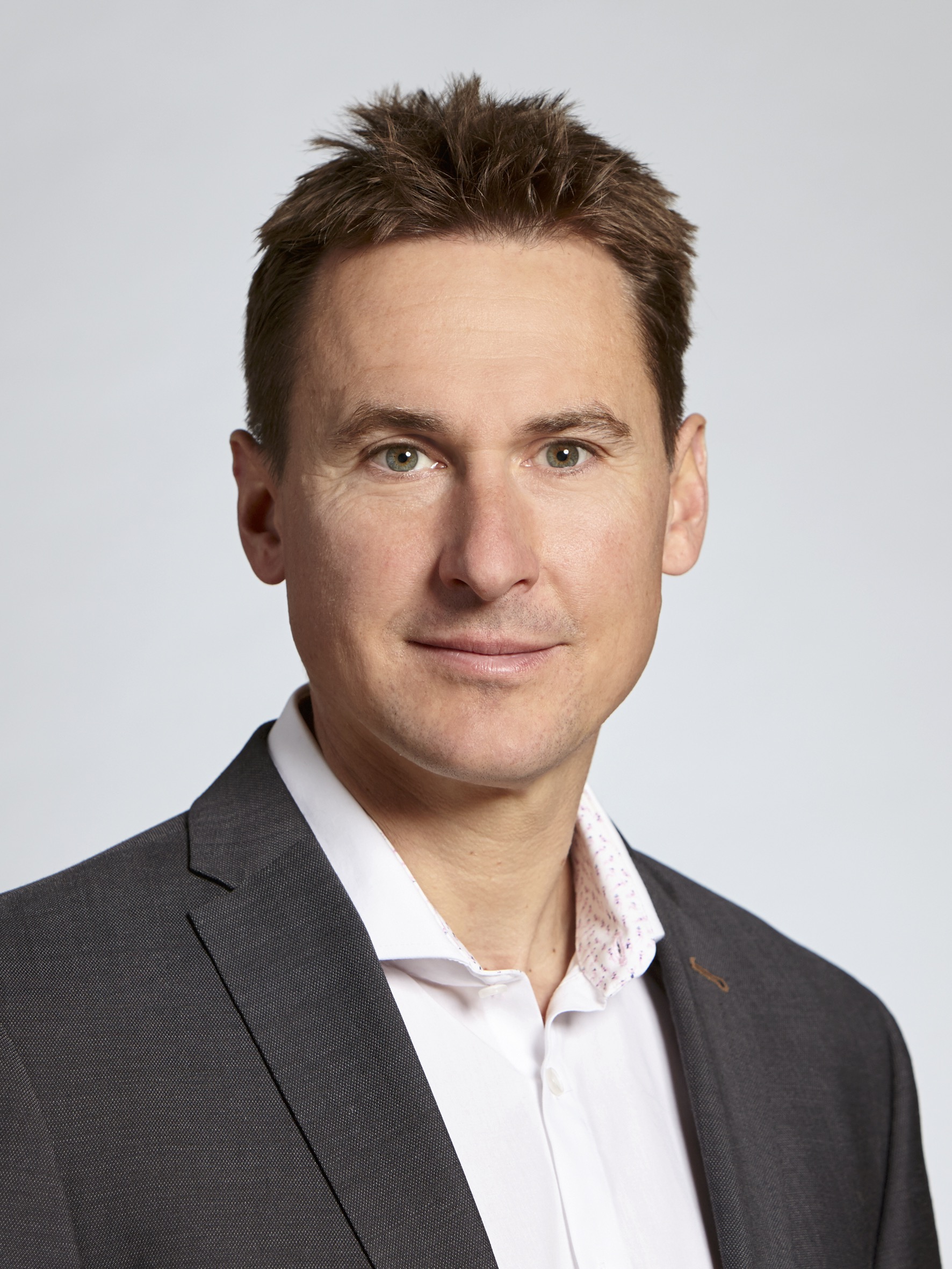}}]{Florian Dörfler}
	is a Full Professor at the Automatic Control Laboratory at ETH Zürich. He received his Ph.D. degree in Mechanical Engineering from the University of California at Santa Barbara in 2013, and a Diplom degree in Engineering Cybernetics from the University of Stuttgart in 2008. From 2013 to 2014 he was an Assistant Professor at the University of California Los Angeles. He has been serving as the Associate Head of the ETH Zürich Department of Information Technology and Electrical Engineering from 2021 until 2022. His	research interests are centered around control, optimization, and system theory with applications in network systems, in particular electric power grids.
\end{IEEEbiography}

\begin{IEEEbiography}[{\includegraphics[width=1in,height=1.25in,clip,keepaspectratio]{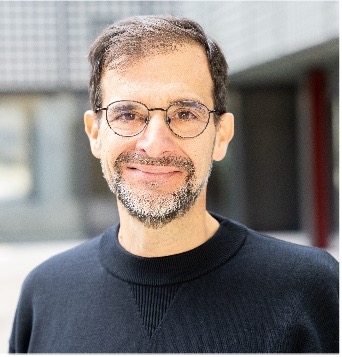}}]{John Lygeros}
	received a B.Eng. degree in 1990 and an M.Sc. degree in 1991 from Imperial College, London, U.K. and a Ph.D. degree in 1996 at the University of California, Berkeley. After research appointments at M.I.T., U.C. Berkeley and SRI International, he joined the University of Cambridge in 2000 as a University Lecturer. Between March 2003 and July 2006 he was an Assistant Professor at the Department of Electrical and Computer Engineering, University of Patras, Greece. In July 2006 he joined the Automatic Control Laboratory at ETH Zurich where he is currently serving as the Head of the laboratory.
\end{IEEEbiography}

\end{document}